\RequirePackage[utf8]{inputenc}
\pdfoutput=1
\documentclass[a4paper,aps,prx,10pt,superscriptaddress,noeprint,twocolumn,nofootinbib]{revtex4-2}

\usepackage[T1]{fontenc}
\usepackage{lmodern}
\usepackage{microtype}

\usepackage{amsmath,amssymb,amsfonts,amsthm,bbm}
\usepackage{bbold}
\usepackage{braket}
\usepackage{mathtools}

\usepackage{bm}
\usepackage{graphicx}
\usepackage[dvipsnames]{xcolor}

\usepackage{booktabs}
\newcommand{\ra}[1]{\renewcommand{\arraystretch}{#1}}
\usepackage{makecell}

\newcommand\myshade{85}
\colorlet{mylinkcolor}{teal}
\colorlet{mycitecolor}{teal}
\colorlet{myurlcolor}{teal}

\usepackage{hyperref}
\hypersetup{
  linkcolor  = mylinkcolor!\myshade!black,
  citecolor  = mycitecolor!\myshade!black,
  urlcolor   = myurlcolor!\myshade!black,
  colorlinks = true,
  breaklinks = true
 }

\newcommand{\id}{\mathbbm{1}}
\newcommand{\idch}{\mathcal{I}}

\newcommand{\I}{\mathrm{i}}

\DeclarePairedDelimiterX{\norm}[1]{\lVert}{\rVert}{#1}
\DeclarePairedDelimiterX{\abs}[1]{\lvert}{\rvert}{#1}

\makeatletter
\let\oldabs\abs
\def\abs{\@ifstar{\oldabs}{\oldabs*}}
\let\oldnorm\norm
\def\norm{\@ifstar{\oldnorm}{\oldnorm*}}
\makeatother

\DeclareMathOperator*{\argmax}{arg\,max}
\DeclareMathOperator*{\argmin}{arg\,min}
\DeclareMathOperator{\Tr}{Tr}
\DeclareMathOperator{\tr}{tr}

\DeclareMathOperator{\spn}{span}

\let\Re\relax
\let\Im\relax
\DeclareMathOperator{\Re}{Re}
\DeclareMathOperator{\Im}{Im}
\renewcommand{\vec}{\mathbf}
\let\originalleft\left
\let\originalright\right
\renewcommand{\left}{\mathopen{}\mathclose\bgroup\originalleft}
\renewcommand{\right}{\aftergroup\egroup\originalright}

\newtheorem{theorem}{Theorem}
\newtheorem{corollary}{Corollary}[section]

\newtheorem{lemmaa}{Lemma}[section]
\newtheorem{theorema}{Theorem}[section]
\theoremstyle{remark}

\newcommand{\npar}{\mathfrak{p}}

\newcommand{\QFI}{\mathcal{F}}
\newcommand{\CQFI}{\mathfrak{F}}
\newcommand{\CQFIbnd}{\mathfrak{B}}
\newcommand{\parvec}{\boldsymbol{\theta}}
\newcommand{\w}{q}  

\newcommand{\rhospace}[1]{\mathcal{S}\left(#1\right)}
\newcommand{\hilb}{\mathcal{H}}

\renewcommand{\t}[1]{\textrm{#1}}

\begin{document}

\title{Probe incompatibility in multiparameter noisy quantum metrology}

\author{Francesco Albarelli}
\affiliation{Faculty of Physics, University of Warsaw, 02-093 Warszawa, Poland}
\author{Rafał Demkowicz-Dobrzański}
\affiliation{Faculty of Physics, University of Warsaw, 02-093 Warszawa, Poland}

\begin{abstract}
We derive fundamental bounds on the maximal achievable precision in multiparameter noisy quantum metrology, valid under the most general entanglement-assisted adaptive strategy, which are tighter than the bounds obtained by a direct use of single-parameter results.
This allows us to study the issue of the optimal probe incompatibility in the simultaneous estimation of multiple parameters in generic noisy channels, while so far the issue has been studied mostly in effectively noiseless scenarios (where the Heisenberg scaling is possible).
We apply our results to the estimation of both unitary and noise parameters, and indicate models where the fundamental probe incompatibility is present.
In particular, we show that in lossy multiple arm interferometry the probe incompatibility is as strong as in the noiseless scenario.
Finally, going beyond the multiple-parameter estimation paradigm, we introduce the concept of \emph{random quantum sensing} and show how the tools developed may be applied to multiple channel discrimination problems.
As an illustration, we provide a simple proof of the loss of the quadratic advantage of time-continuous Grover algorithm in presence of dephasing or erasure noise.
\end{abstract}

\maketitle

\section{Introduction}

Precise characterization of parameters of physical systems is an important task, both from a technological as well as a purely scientific perspective.
Understanding the limits on how precisely the parameters can be estimated, given the most general estimation protocols admitted by quantum mechanics,
touches also upon the foundations of quantum mechanics itself.
When it comes to \emph{single parameter} quantum metrology, both the theory and practical applications are now in their maturity stage.
The theory, not only provides the fundamental bounds that indicate in which models one may expect the most promising quantum enhancements~\cite{Escher2011, Demkowicz-Dobrzanski2012, Demkowicz-Dobrzanski2014, Demkowicz-Dobrzanski2017, Zhou2017}, but also provides explicit protocols based on the use of squeezed states and quantum error-correction ideas to reach these limits~\cite{Escher2011, Demkowicz2013, Zhou2017, Zhou2019e, Zhou2020}.
The most spectacular application is the use of squeezed states of light in modern gravitational wave detectors~\cite{LIGO2019,Virgo2019}, which operate surprisingly close to the fundamental limits~\cite{Demkowicz2013}, taking into account the level of optical losses present in the devices.

\emph{Multiparameter} estimation problems are abundant, e.g. vector field estimation~\cite{Baumgratz2015}, multi-arm interferometry~\cite{Humphreys2013}, waveform estimation~\cite{Tsang2011}, etc., and this research domain has rightfully attracted an increasing amount of attention in recent years~\cite{Szczykulska2016,Albarelli2019c}.
The aim of multiparameter quantum metrology is to obtain the most precise estimates of several parameters \emph{simultaneously}, i.e. within a single experimental configuration.
A quantum metrology experiment can schematically be divided in three stages: the preparation of a probe state, the sensing stage when the probe's evolution is affected by the parameters of interest and finally the measurement.
Assuming that the evolution of the probe is  fixed by the physical nature of the problem, the final goal is to find a combination of probe state and measurement that achieves the best possible precision of estimating unknown parameters of the evolution.
As such, a metrological problem may be regarded formally as a quantum channel estimation problem.
The multiparameter character, however, adds another layer of complexity on top of single parameter scenarios; identification of fundamental bounds and optimal protocols becomes much more challenging.

In many practical situations the channel can be probed many times and the actual optimal protocols may be adaptive.
Evaluating the power of adaptive strategies, especially in the presence of a noisy environment, is a challenging task appearing in various contexts throughout the whole field of quantum information theory~\cite{Terhal2015, Pirandola2017, Zhuang2020b}.
It is remarkable that this problem has been completely resolved in the case of single-parameter quantum metrology~\cite{Demkowicz-Dobrzanski2014, Demkowicz-Dobrzanski2017, Zhou2017, Zhou2019e, Zhou2020}, proving the effectiveness of the theoretical methods developed.
The main goal of this paper is to generalize these methods to the multiparameter scenario.

\begin{figure}[t!]
\includegraphics[width=.97\columnwidth]{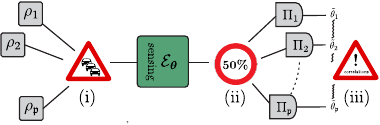}
\caption{Three potential sources of incompatibility in multiparameter quantum metrology:
(i) different input probe states may be optimal to estimate different parameters---this may in the worst case lead to a $\npar$-fold increase in the resources consumed, compared with best case where a single probe state is optimal for estimating all parameters;
(ii) different incompatible measurements may be optimal for extracting information about different parameters---in generic noisy metrological models this will require at most double the required resources in the asymptotic regime, where many probes may be used;
(iii) if the chosen parametrization is not ``natural'' for the problem at hand the resulting estimators will manifest correlations, and the imperfect knowledge of some parameters will have an impact on the effective estimation uncertainty of the remaining ones.
}
\label{fig:intro}
\end{figure}

The most intriguing aspect of multiparameter quantum metrology is the existence of protocols for simultaneous estimation that achieve a better overall precision than estimating each parameter separately, given the same amount of resources (e.g. number of particles, total sensing time).
Indeed, several theoretical~\cite{Humphreys2013,Baumgratz2015,Yuan2016b,Gessner2018,Kura2017,Li2019a,Gessner2020} and experimental~\cite{Polino2019,Hou2020,Hou2021,Hou2021a} studies have shown such an advantage in noiseless and error-corrected~\cite{Gorecki2020} scenarios.
In the best case, it may be possible to estimate all the parameters in a single experiment with the same precision obtainable in separate experiments for each parameter.
The conditions for such a maximal advantage, known as \emph{compatibility} conditions, have been laid out in~\cite{Ragy2016}.
We report them here: \emph{``(i) existence of a single probe state allowing for optimal sensitivity for all parameters of interest, (ii) existence of a single measurement optimally extracting information from the probe state on all the parameters, and (iii) statistical independence of the estimated parameters.''}
These three aspects are pictorially represented in Fig.~\ref{fig:intro}.

The impossibility to satisfy condition (ii) is known as \emph{measurement} incompatibility, and it boils down to the fact that the optimal observables to estimate different parameters might not commute.
This issue has always been central in quantum estimation theory~\cite{helstrom1976quantum,Holevo2011b,Paris2009,Liu2019d,Demkowicz-Dobrzanski2020}: from the seminal studies of almost half a century ago~\cite{Yuen1973,Belavkin1973,Holevo1976} to recent developments~\cite{Matsumoto2002,Pezze2017,Albarelli2019,Yang2018b,Sidhu2019a,Razavian2020,Lu2020a,Conlon2020,Belliardo2021}.
An important and relatively new observation is that measurement incompatibility will at most double the total mean squared error on the parameters' estimates~\cite{Carollo2019,Tsang2019} when a large number of identical copies of the probe state can be measured collectively~\cite{Kahn2009,Yamagata2013,Yang2018a}.
Importantly, for the generic noisy protocols that are the main focus of this paper, the maximal quantum enhancement amounts to a constant factor gain and the optimal probe states may be effectively approximated by product states of finitely entangled groups of particles~\cite{Jarzyna2013,Chabuda2020}.
As such, the argument on the impact of measurement incompatibility being at most factor of two applies also to the asymptotically optimal strategies for noisy metrology.

Out of the three compatibility conditions (i-iii) stated above, the most relevant one is actually condition (i), being the only one responsible for a different scaling of the asymptotic precision with the number of parameters involved, when comparing optimal protocols for simultaneous estimation with separate ones.
We will refer to violation of condition (i) as \emph{probe} incompatibility\footnote{For convenience we keep the same name also when considering more general probing strategies with multiple uses of the channel.} and it will be the main focus of our analysis.
When it comes to condition (iii), statistical independence of parameters can always be assured by a proper reparametrization, hence, provided one starts with a ``natural'' parametrization for a given model, this aspect of incompatibility can be avoided.

In this paper, we focus primarily on generic noisy channels, where the noise cannot be completely removed without hindering the parameter encoding process.
In this case, the asymptotic precision follows the so called standard quantum limit (SQL) and quantum-enhanced strategies will provide at most a constant gain~\cite{Escher2011,Demkowicz-Dobrzanski2012}.
Such noisy models, while generic and ubiquitous in practice, are much less studied in the multiparameter literature.
A few particular instances have been examined, e.g.~\cite{Tsang2013a,Yue2014,Vidrighin2014,Crowley2014,Ho2020,Friel2020}, but theoretical tools to identify fundamental bounds without neglecting the multiparameter character of the problem, especially probe incompatibility, are missing.
We aim to close this gap.

\subsection{Summary of results}

First, in Sec.~\ref{sec:multiQEst} we set the stage by defining a new figure of merit to quantify incompatibility in multiparameter quantum metrology.
While this figure of merit takes into account all conditions (i)-(iii), see Fig.~\ref{fig:intro}, in the rest of the paper we focus on lower bounds that consider only condition (i), probe incompatibility.

To this end, in Sec.~\ref{sec:multiQFIsum} we derive a new class of multiparameter precision bounds that hold for the most general adaptive strategy depicted in Fig.~\ref{fig:schemes}(a), extending previous single-parameter results~\cite{Fujiwara2008,Escher2011,Demkowicz-Dobrzanski2012,Demkowicz-Dobrzanski2014}.
In particular, we derive bounds that apply to various scenarios, as summarized in Table~\ref{tab:bounds}.
These represent our main technical contribution and, conveniently, the optimal bounds in this class can be evaluated with semidefinite programs, presented in Appendix~\ref{app:SDP}.

The idea behind the derivation is simple but powerful and for a single-parameter it has proven to be the most powerful and widely applicable approach.
A noisy channel can always be purified (formally) as a unitary interaction between the system and an inaccessible environment, by choosing purifications that contain as little information as possible about the parameters one can obtain tight bounds on the metrological precision.
Thanks to a few technical adjustments, existing single-parameter derivations can be extended to multiple parameters in a way that takes into account probe incompatibility.
Actually, these bounds are more general and apply also to a scenario that we call random quantum sensing, depicted in Fig.~\ref{fig:schemes}(b), in which different channels chosen at random must be probed by the same state.

After introducing these tools, we then apply them to study probe incompatibility for a few paradigmatic models of noisy multiparameter quantum metrology in Sec.~\ref{sec:applications}.
Finally, in Sec.~\ref{sec:discrimination} we show how multiparameter metrological bounds can be used to assess the ultimate performance of adaptive strategies in the task of discriminating between multiple quantum channels, using a geometric argument sketched in Fig.~\ref{fig:schemes}(c).
This approach allows us to close a conjecture presented in Ref.~\cite{Demkowicz-Dobrzanski2015} and show that the quantum computational speed-up of Grover search is ruined by dephasing and erasure noise.

\subsubsection*{Relation to previous works}

A multiparameter bound for noisy channels and adaptive strategies was recently obtained without relying on purification arguments~\cite{Katariya2020b}.
Despite being simple to evaluate, it is both less general and less tight than the optimal purification-based bounds we introduce here\footnote{Indeed, in Appendix~\ref{app:RLD} we show that the class of bounds we introduce includes also the one of~\cite{Katariya2020b}, which is generally suboptimal and might not give a complete insight into probe incompatibility.}.
Moreover, a purification-based methodology was introduced in~\cite{Chen2017a}; yet, by construction it does not take into account probe incompatibility, since it requires performing a different convex optimization for each parameter.
Finally, multiparameter bounds for paradigmatic models in optical metrology with losses have been obtained by purifying the dynamics~\cite{Tsang2013a,Yue2014}, but without considering the possibility of adaptive strategies.

\section{Multiparameter quantum estimation}
\label{sec:multiQEst}

\begin{figure}
\includegraphics[width=.97\columnwidth]{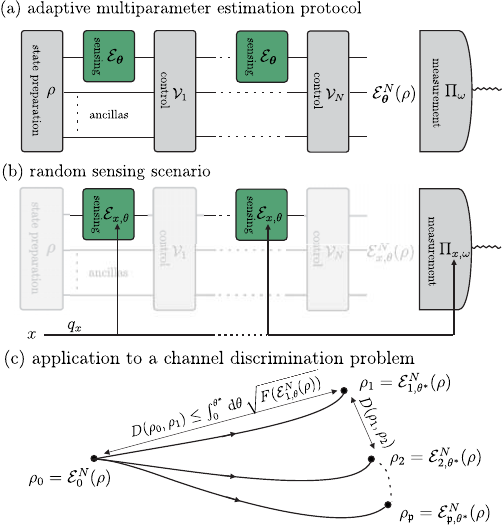}
\caption{Schematic representation of the three operational tasks considered in this paper:
(a) estimation of the parameters $\parvec=(\theta_1, \dots, \theta_\npar)$ appearing in the quantum channel $\mathcal{E}_{\parvec}$, under the most general entanglement-assisted adaptive strategy with $N$ uses of the channel; (b) in random quantum sensing a single parameter $\theta$ is encoded by a channel $\mathcal{E}_{x,\theta}$ randomly chosen from $\npar$ channels with probability $q_x$, the extracted $x$ is revealed only before the measurement stage; (c) intuitive representation of the geometrical argument of Sec.~\ref{sec:discrimination} connecting the error in the discrimination of $\npar$ quantum channels with the precision of random sensing.
}
\label{fig:schemes}
\end{figure}

\subsection{Quantum Fisher information matrix}

We consider a vector $\parvec=[\theta_1,\dots,\theta_\npar]^T$ of $\npar$ real parameters that are encoded on a quantum state $\rho_{\parvec}=\mathcal{E}_{\parvec}(\rho)$ via the parameter-dependent quantum channel $\mathcal{E}_{\parvec}$ acting on the initial state $\rho$.
In this work we consider only finite-dimensional systems.

A generic measurement is described by a positive operator valued measure (POVM), i.e. a set of positive operators $\Pi_\omega \geq 0$ satisfying $\sum_\omega \Pi_\omega = \id$.
We take the set of possible outcomes $\omega$ to be finite-dimensional without loss of generality.
The parameter-dependent classical probability distribution is obtained from the Born rule $p_{\parvec}(\omega)=\Tr \left[ \Pi_\omega \rho_{\parvec} \right]$.

An estimator $\tilde{\parvec}(\omega)$ is a function that maps the random outcomes to estimated parameters values.
The precision of the estimator is quantified by the mean square error matrix:
\begin{equation}
\Sigma_{\tilde{\parvec}} \coloneqq \sum_\omega  p_{\parvec}(\omega) [ \tilde{\parvec}(\omega) - \parvec] [ \tilde{\parvec}(\omega) - \parvec]^T.
\end{equation}
In particular, we consider locally unbiased estimators that satisfy $\sum_x p_{\parvec^*}(\omega) \tilde{\parvec}(\omega) = \parvec^*$ and $\sum_\omega  \partial_{\theta_i} p_{\parvec}(\omega) \mid_{\parvec = \parvec^*} \tilde{\theta}_j(\omega) = \delta_{ij}$, i.e. they are unbiased locally around the true value $\parvec^*$ of the parameter vector.
In the following we will implicitly take all derivatives wrt the parameters $\theta_j$ evaluated at the true value and make the true value explicit only if needed.
For this class of estimators the mean square error matrix is equal to the covariance matrix and for any POVM it satisfies the quantum Cramér-Rao bound (QCRB)~\cite{helstrom1976quantum,Holevo2011b,Paris2009,Liu2019d}
\begin{equation}
\label{eq:matrixQCRB}
\Sigma_{\parvec} \geq \QFI^{-1}(\rho_{\parvec}),
\end{equation}
i.e. $\Sigma_{\parvec}-\QFI^{-1}(\rho_{\parvec})$ is positive semidefinite, where we have introduced the quantum Fisher information (QFI) matrix
\begin{equation}
\label{eq:QFImatdef}
\QFI_{ij}(\rho_{\parvec}) \coloneqq \Re (\Tr \left[ \rho_{\parvec} L_i L_j \right]),
\end{equation}
written in terms of the symmetric logarithmic derivatives (SLDs), defined by the equations
\begin{equation}
	\partial_{\theta_i} \rho_{\parvec} = \frac{1}{2}\left(L_i \rho_{\parvec} + \rho_{\parvec} L_i\right).
\end{equation}
For single-parameter problems ($\npar=1$) it is always possible to find a POVM and a locally unbiased estimator to attain the QCRB
(formally, at a given operation point $\theta^*$, or, more practically, uniformly in the asymptotic limit~\cite{Hayashi2005}).
A possible choice for the optimal measurement is a projective measurement on the eigenbasis of the SLD operator.
The multiparameter matrix bound~\eqref{eq:matrixQCRB}, however, is not always attainable due to the possible non-commutativity of the SLDs \cite{Hayashi2005, Holevo2011b, Demkowicz-Dobrzanski2020}.

\subsection{Lower bounds on the total variance}

In an experiment that aims to estimate multiple parameters, in general there exists no strategy yielding a minimal covariance matrix~\cite{Silvey1980}.
A common choice is to quantify the overall error with a scalar, which we call the \emph{total variance}
\begin{equation}
\label{eq:totvariance}
\Delta_W^2 \tilde{\parvec} \coloneqq \tr[ W \Sigma_{\tilde{\parvec}} ],
\end{equation}
where $W>0$ is a positive cost matrix (when $W=\id_\npar$ we simply write $\Delta^2 \tilde{\parvec}$).
Every cost matrix can be decomposed as $W= \sum_{k=1}^{\npar} w_k \vec{e}_k \vec{e}_k^T$, where $\vec{e}_k$ are orthonormal vectors and $w_k > 0$ are strictly positive weights.
A choice of cost matrix distinguishes a particular set of parameters, associated with the eigenvectors of $W$,for which the estimation cost is determined by the corresponding eigenvalues $w_k$.
In a sense, the cost matrix determines which parameters we regard as ``separate''.
In what follows we assume to be working in the parametrization induced by $W$, taking into account only the effect of the weights.

From the matrix QCRB~\eqref{eq:matrixQCRB} we can lower bound the total variance as
\begin{equation}
\label{eq:JinvBound}
\Delta_{W}^2 \tilde{\parvec} \geq \sum_{x=1}^{\npar} w_x [\QFI^{-1}]_{xx} \geq \sum_{x=1}^\npar \frac{w_x}{\QFI_{xx}}  \geq \frac{\npar^2}{\sum_{x=1}^{\npar} w_x^{-1} \QFI_{xx}  },
\end{equation}
where we have written the Fisher information matrix elements $\QFI_{xx}$ in the $W$ eigenbasis.

The first inequality $\Delta_{W}^2 \tilde{\parvec} \geq  \tr W \QFI^{-1}$ is generally not attainable due to measurement incompatibility.
A more fundamental bound is the Holevo Cramér-Rao bound, which is asymptotically attainable \cite{Yang2018a,Demkowicz-Dobrzanski2020}.
However, the Holevo Cramér-Rao bound is at most $2 \tr W \QFI^{-1}$~\cite{Tsang2019}, therefore the asymptotic effect of incompatibility on the total variance is at most a factor 2.
As argued in the introductory section, in this paper we are not concerned with measurement incompatibility and we will only consider bounds obtained from the QFI matrix.

The second inequality in~\eqref{eq:JinvBound} is a property of positive matrices and it means that statistical correlations among the parameters, i.e. a non-diagonal $\QFI$, increase the error on the $x$-th parameter with respect to the single-parameter QCRB $1/\QFI_{xx}$, i.e. assuming all the other parameters are known.
This inequality will be tight, provided the parameter basis choice induced by the cost matrix $W$ coincides with the parameter basis in which the QFI matrix is diagonal (we will refer to this as a `natural' parametrization).
The third inequality is obtained from the Cauchy–Schwarz inequality (or equivalently from the inequality between harmonic and arithmetic mean)
and is saturated when the diagonal elements are all equal.
In multiparameter quantum metrology, the weaker bound~\eqref{eq:JinvBound} has often been used to avoid computing the inverse and simplify calculations, see e.g.~\cite{Ballester2004a,Kura2017,Ge2017,Katariya2020b}.

The lower bound~\eqref{eq:JinvBound} prompts us to consider the quantity
\begin{equation}
\CQFI_{\vec{\w}}(\rho) \coloneqq \sum_{x=1}^\npar w_x^{-1} \QFI\left( \mathcal{E}_{\theta_x} ( \rho ) \right) =  \sum_{x=1}^\npar \w_x
\QFI\left(\mathcal{E}_{\theta_x}(\rho) \right),
\end{equation}
which we call the \emph{total} QFI, with weights $\w_x = w_x^{-1}$. For the sake of a greater generality, and anticipating the application in quantum channel discrimination problems discussed in Sec.~\ref{sec:discrimination}, we will be using the same notation also when, instead of considering a single channel with multiple parameters, we will consider a single parameter being encoded using \emph{different} quantum channels.

In this second case, the weight $\w_x$ corresponds to the probability (up to renormalization) that a given channel has acted on the state
and the total QFI has an operational meaning in a scenario that we will name  \emph{random quantum sensing}, schematically represented in Fig.~\ref{fig:schemes}-(b).
In this scenario, Alice sends a probe state $\rho$ through a $\theta$-dependent channel $\mathcal{E}_{x,\theta}$, selected randomly according to the $\theta$-independent probability $\w_x$ from the ensemble of $\npar$ channels.
Alice does not know which channel $x$ is selected and has to choose a unique state $\rho$ to send in all runs.
At the other end of the channel, Bob knows both Alice's probe state $\rho$ and the random value $x$, so each time he can implement an optimal measurement and estimator for the parameter $\theta$.
Alice's goal is to help Bob estimate $\theta$ as precisely as possible, therefore she has to prepare a state that is sensitive to $\theta$ for all the different channels, according to their probability.
The quantity $\CQFI_{\vec{\w}}$ quantifies the precision that Bob obtains in estimating $\theta$ when Alice chooses to send $\rho$.
This task, though different, is reminiscent of the quantum random access codes protocols~\cite{Ambainis2009,Tavakoli2015a}, in a sense that the sender must prepare a quantum state so that the information is extracted optimally even though it is a priori not clear how the information will be encoded (random quantum sensing) or
which bit of information will be read out (random access codes).
Note that in the random sensing scenario the total QFI leads by construction to a tight bound on the estimation precision of $\theta$, without the additional saturability issues affecting the standard multiparameter setting indicated in~\eqref{eq:JinvBound}.

\subsection{Incompatibility measures}
\label{sec:incomp_meas}

Here we introduce a quantity that captures the incompatibility of a given multiparameter quantum channel estimation problem.
We aim at a quantity that is $1$ whenever there are no incompatibility issues, and increases to some maximum value when the optimal probe for estimating a given parameter gives no information on the remaining ones.
Moreover, we would like the quantity not to depend on the choice of the cost matrix nor on the particular parametrization.
The following natural quantity which we refer to as the \emph{incompatibility measure of a multiparameter quantum channel} satisfies these requirements:
\begin{equation}
\label{eq:incompmeasure}
\mathfrak{I}^*(\mathcal{E}_{\parvec}) \coloneqq \max_{\{\vec{w}_x\}} \left(\frac{\min
\limits_{ \rho, \Pi , \tilde\parvec }\Delta_W^2 \tilde{\parvec}}{\sum\limits_x \min\limits_{\rho_x} \vec{w}_x^T \QFI^{-1}(\mathcal{E}_{\parvec}(\rho_x)) \vec{w}_x} \right),
\end{equation}
where $\tilde \parvec$ are locally unbiased estimators\footnote{Equivalently we could define this quantity in terms of the classical Fisher information matrix associated to a POVM, without explicit reference to estimators.} and the arbitrary (not necessarily orthogonal) vectors $\vec{w}_x$ determine the effective cost matrix for the multiparamter estimation problem $W = \sum_x \vec{w}_x  \vec{w}_x^T$.
Vectors $\vec{w}_x$ may be understood as representing a certain scalar function of $\parvec$ to be estimated, corresponding to a choice of a rank-1 cost matrix $W_x = \vec{w}_x  \vec{w}_x^T$.
For a given $\rho_x$ the quantity $\vec{w}_x^T \QFI^{-1}(\mathcal{E}_{\parvec}(\rho_x)) \vec{w}_x$, appearing in the denominator, is an attainable bound on the minimal error of estimating the function when the remaining parameters are treated as nuisance parameters~\cite{Tsang2019,Suzuki2019a} (due to effectively scalar nature of the estimation problem,  measurement incompatibility does not affect attainability of the bound in this case).

Intuitively, the incompatibility measure captures the worst case scenario where the ratio between the cost of estimating $\npar$ scalar functions simultaneously is the largest compared with the cost of estimating them separately with the optimal probes.
Note that $\mathfrak{I}^*$ is parametrization independent, since for any reparametrization $\parvec \rightarrow A \parvec$ (where $A$ is some invertible matrix) we will obtain the identical formula by the following change of the cost vectors  $\vec{w}_x \rightarrow A^T \vec{w}_x$.
Moreover, this quantity is indeed $1$ when there is a single state $\rho$ that gives minimal cost for all $x$.

Admittedly, this quantity is challenging to compute, especially for multiparameter metrological models in presence of noise.
Therefore, in this paper we will be using an efficiently computable lower bound, the \emph{probe incompatibility measure} based on the total QFI:
\begin{equation}
\label{eq:probe_inc_cost}
\mathfrak{I}^* (\mathcal{E}_{\parvec}) \geq \mathfrak{I}(\mathcal{E}_{\parvec}) \coloneqq \npar \left( \max_{\rho} \sum_{x=1}^\npar \frac{ \QFI_{xx} \left( \mathcal{E}_{\parvec}( \rho ) \right) }{ \max_{\rho_x} \QFI_{xx} \left( \mathcal{E}_{\parvec} ( \rho_x )  \right) } \right)^{\!\!-1}
\end{equation}
where some `natural' parametrization has been fixed---see Appendix~\ref{app:parametrizationbound} for the derivation of the bound.
This way we obtain a quantity that takes values in the range $1 \leq \mathfrak{I}(\mathcal{E}_{\parvec})  \leq \npar $, where again $1$ indicates perfect compatibility and $\npar$ maximal incompatibility (when the best strategy is to estimate the parameters separately and the state that maximizes the QFI for a given $x$ yields zero QFI for all other $x$).
In what follows, instead of
$\max_{\rho_x} \QFI_{xx} \left( \mathcal{E}_{\parvec} ( \rho_x )\right)$ we will also write
$\max_{\rho_x} \QFI \left( \mathcal{E}_{\theta_x} ( \rho_x ) \right)$, with a single scalar parameter $\theta_x$, to indicate that this quantity refers to essentially a single parameter problem, where all other parameters can be regarded as perfectly known.

Unlike the original incompatibility measure defined in Eq.~\eqref{eq:incompmeasure}, the lower bound~\eqref{eq:probe_inc_cost} is in general parametrization-dependent (though it is invariant under rescaling of parameters).
For the lower bound~\eqref{eq:probe_inc_cost} to hold, one must choose a `natural' parametrization, such that the QFI matrices corresponding to the optimal states $\argmax_{\rho_x} \QFI_{xx} \left( \mathcal{E}_{\parvec} ( \rho_x )  \right)$ are diagonal.
Nonetheless, we argue that $\mathfrak{I}(\mathcal{E}_{\parvec})$ is meaningful to study on its own for any given parametrization.
While $\mathfrak{I}^*(\mathcal{E}_{\parvec})$ takes into account all conditions (i)-(iii) stated in the introductory section, the probe incompatibility measure $\mathfrak{I}(\mathcal{E}_{\parvec})$ singles out the effect of (i), as the name suggests.
In principle, one can check a posteriori if the QFI matrices of the optimal states are diagonal to gauge the `naturalness' of the chosen parametrization.

Finally, let us note, that when the probe incompatibility measure is applied to the random sensing scenario, the quantity~\eqref{eq:probe_inc_cost} is actually the true incompatibility cost and not a lower bound, as there is a clear distinction between the parameters of different channels, unlike in the genuine multiparameter scenario.

\section{Bounds on the total QFI}
\label{sec:multiQFIsum}

In this section we study the total QFI of a collection of $\npar$ quantum channels (i.e. completely-positive trace-preserving maps) $\mathcal{E}_{x,\theta_x} (\rho)=\sum_{j=1}^{r_x} K_{x,j} \rho K^{\dag}_{x,j}$, labeled by $x=1\,\dots,\npar$, each with a dependence on a parameter of interest $\theta_x$.
To ease the notation we will mostly suppress the label $x$ and identify the different channels only through their parameter $\theta_x$ as $\mathcal{E}_{\theta_x}$.
Depending on the physical context, the resulting bounds will have applications either in the standard multiparameter estimation setting or in the random sensing scenario.
We consider channels that act on operators on the same input Hilbert space $\hilb$, but the outputs can have different dimensions.
When needed we  will use a shorthand notation for channels acting on pure states $\mathcal{E}(\ket{\psi})\equiv\mathcal{E}(|\psi\rangle\langle \psi|)$.

In this section, we derive an attainable bound for a single use of the channel and an upper bound for generic strategies with $N$ uses of the channel.
In Appendix~\ref{app:previous} we also show that some previously known results, including the bound of~\cite{Katariya2020b}, can be obtained within our approach.
For practical applications it is crucial to note that all the bounds we derive (the single-use bound $\CQFI_{\vec{\w}}$~\eqref{eq:singleuseCEbound}, and the asymptotic channel bound $\CQFIbnd_{\vec{\w}}$~\eqref{eq:SQLbound}) can be computed with semidefinite programs (SDPs), similarly to their single-parameter counterparts \cite{Demkowicz-Dobrzanski2012,Koodynski2013,Demkowicz-Dobrzanski2017}; see  Appendix~\ref{app:SDP}.
A MATLAB implementation can be found in~\cite{githubrepo}.
A summary of all the bounds derived in this paper is shown in Table~\ref{tab:bounds}.

\begin{table}
\centering
\ra{1.3}
\begin{tabular}{ccc}
\toprule
Case & Eq. & SDP \\
\midrule
\makecell{Single channel use ($N=1$)} & \eqref{eq:singleuseCEbound} & Yes \\
\makecell{Parallel strategy, finite-$N$} & \eqref{eq:parallelCEbound} & Yes \\
\makecell{Adaptive strategy, finite-$N$} & \eqref{eq:seqCEbound} & No \\
\makecell{ Parallel \& adaptive,  asymp. SQL} & \eqref{eq:SQLbound} & Yes \\
\makecell{Markovian noise, adaptive
\\ asymp. SQL (in $T$)} & \eqref{eq:SQLboundLind} & Yes \\
\bottomrule
\end{tabular}
\caption{Summary of the bounds on the total QFI obtained in this paper, the SDPs can be found in Appendix~\ref{app:SDP}.}
\label{tab:bounds}
\end{table}

\subsection{Single use of the channel}
Here we focus on the total channel QFI, i.e. the optimal total QFI:
\begin{equation}
\CQFI_{\vec{\w}} \coloneqq \max_\rho \CQFI_{\vec{\w}}(\rho) = \max_{\rho} \sum_{x=1}^\npar \w_x \QFI \left( \mathcal{E}_{\theta_x} (\rho) \right).
\end{equation}
One can also consider a more general entanglement-assisted strategy, corresponding to the extended channels $\mathcal{E}_{\theta_x} \otimes \idch_A$ ($\idch_A$ denotes the identity channel on an auxiliary Hilbert space) so that the maximization runs over bipartite states $\rho \in \rhospace{\hilb_S \otimes \hilb_A}$, where $\rhospace{\hilb}$ denotes the space of density matrices over $\hilb$.
The total channel QFI for this extended channel is always an upper bound on $\CQFI_{\vec{q}}$.
In this work we do not investigate the difference between assisted an unassisted strategies, therefore the distinction between the two cases is not crucial and will not be made explicit unless necessary, even if the bounds are actually derived for the extended channels.

The most straightforward way to upper bound the total channel QFI is to use the single-parameter channel QFI $\mathfrak{F}_x$ for each parameter:
\begin{equation}
\label{eq:tria_ineq_single_use}
\CQFI_{\vec{\w}} \leq \sum_{x=1}^\npar \w_x \max_{\rho_x}  \QFI \left( \mathcal{E}_{\theta_x}(\rho_x)\right) = \sum_{x=1}^\npar \w_x \mathfrak{F}_x.
\end{equation}
This upper bound, however, does not take into account potential probe incompatibility and indeed when it is attained  the incompatibility cost takes its minimal value $\mathfrak{I}(\mathcal{E}_{\parvec}) = 1$.

In order to obtain a tighter bound we revisit and generalize the derivation of the single-parameter channel QFI~\cite[Th.~4]{Fujiwara2008},
which results in the following.
\begin{theorem}\label{theo:multiparFujiImaiGen}
The total channel QFI of a collection of quantum channels is upper bounded as
\begin{align}
\label{eq:singleuseCEbound}
\CQFI_{\vec{\w}}  \leq
4 \min_{ \mathfrak{h} }  \norm{ \sum_{x=1}^\npar \w_x \alpha_x }
 ,
\end{align}
where $\norm{ \cdot }$ denotes the operator norm (maximal singular value), $\alpha_x \coloneqq \partial_{\theta_x}\!\tilde{\vec{K}}_x^\dag \partial_{\theta_x}\!\tilde{\vec{K}}_x$ with $ \partial_{\theta_x}\!\tilde{\vec{K}}_x \coloneqq  \partial_{\theta_x}\!\vec{{K}}_x-\I h_x \vec{K}_x$  and $\mathfrak{h}=\{ h_x \}_{x=1}^{\npar}$ is a collection of $\npar$ Hermitian matrices, each of dimension $r_x {\times} r_x$.
Equality in~\eqref{eq:singleuseCEbound} is attained when considering the extended channels $\mathcal{E}_{\theta_x} \otimes \idch_A$.
\end{theorem}
\begin{proof}
We consider the extended channels $\mathcal{E}_{\theta_x} \otimes \idch_A$, which will give an upper bound for the unextended ones.
First, we note that the total QFI is maximized by a pure state $\ket{\psi} \in \hilb_S \otimes \hilb_A$, because of its convexity.
Each single-parameter QFI in the sum can be written as a minimization over an Hermitian matrix $h_x$~\cite{Fujiwara2008,Demkowicz-Dobrzanski2012}:
\begin{equation}
\label{eq:QFIextpuremix}
\begin{split}
\QFI \left( \mathcal{E}_{\theta_x} \otimes \idch ( \ket{\psi} ) \right) &= 4\min_{ h_x }  \bra{\psi} \partial_{\theta_x}\!\tilde{\vec{K}}_x^\dag \, \partial_{\theta_x}\!\tilde{\vec{K}}_x \otimes \id \ket{\psi} \\
&= 4\min_{ h_x }  \Tr \left[ \rho \, \partial_{\theta_x}\!\tilde{\vec{K}}_x^\dag \, \partial_{\theta_x}\!\tilde{\vec{K}}_x \right],
\end{split}
\end{equation}
where $\rho = \Tr_A |\psi \rangle \langle \psi|$.
Therefore it is equivalent to maximize the lhs of Eq.~\eqref{eq:QFIextpuremix} over pure bipartite states and the rhs over the convex set of mixed states $\rhospace{\hilb}$.
The minimum of the sum of functions of the independent variables $h_x$ is equal to the sum of the minima; thus we obtain
\begin{align}
& \max_{ \ket{\psi} } \sum_{x=1}^\npar \w_x \QFI \left( \mathcal{E}_{\theta_x} \otimes \idch (\ket{\psi}) \right) \\
& = 4 \max_{ \rho \in \mathcal{S}(\mathcal{H}) } \min_{\mathfrak{h} } \Tr \left[ \rho \sum_{x=1}^\npar \w_x \partial_{\theta_x}\!\tilde{\vec{K}}_x^\dag \, \partial_{\theta_x}\!\tilde{\vec{K}}_x \right] \label{eq:maxmin}\\
&= 4 \min_{ \mathfrak{h} }  \max_{ \rho \in \mathcal{S}(\mathcal{H}) } \Tr \left[ \rho \sum_{x=1}^\npar \w_x \partial_{\theta_x}\!\tilde{\vec{K}}_x^\dag \, \partial_{\theta_x}\!\tilde{\vec{K}}_x \right] \\
&=4 \min_{ \mathfrak{h} }  \norm{ \sum_{x=1}^\npar  \w_x \partial_{\theta_x}\!\tilde{\vec{K}}_x^\dag \, \partial_{\theta_x}\!\tilde{\vec{K}}_x }.
\end{align}
The conditions to interchange max and min in~\eqref{eq:maxmin} are the convexity and compactness of the set $\mathcal{S}(\mathcal{H})$, the concavity of the cost function in $\rho$ and its convexity in the variables $h_x$~\cite{Rockafellar1970}[corollary 37.3.2].
These conditions are satisfied, since the function is linear in $\rho$, each term in the sum is convex in the matrix $h_x$ and the sum of convex functions with positive coefficients is convex.
\end{proof}
Since the bound~\eqref{eq:singleuseCEbound} is the norm of a sum of operators, Eq.~\eqref{eq:tria_ineq_single_use} immediately follows from the triangle inequality, where $\mathfrak{F}_x = 4 \min_{ h_x }  \norm{ \partial_{\theta_x}\!\tilde{\vec{K}}_x^\dag \, \partial_{\theta_x}\!\tilde{\vec{K}}_x }$ is the single-parameter extended channel QFI~\cite{Fujiwara2008}.
We will use the bound~\eqref{eq:singleuseCEbound} with weights $\w_x=1/\mathfrak{F}_x$ to evaluate the probe incompatibility cost~\eqref{eq:probe_inc_cost} of the (extended) channel.

\subsection{Asymptotic bound for the most general strategy}

Now we move to a metrological strategy with $N$ uses of the channel $\mathcal{E}_{\theta_x}$.
We consider the most general adaptive strategy, allowing for arbitrary auxiliary systems and unitary control operations $\mathcal{V}_i ( \rho ) = V_i \rho V_i^\dag$, as depicted in Fig.~\ref{fig:schemes}-(a)
The overall channel is thus
$\mathcal{E}^N_{\theta_x}=\mathcal{V}_N \circ (\mathcal{E}_{\theta_x}\otimes\idch) \circ \mathcal{V}_{N-1} \circ  \dots  \circ \mathcal{V}_1 \circ (\mathcal{E}_{\theta_x}\otimes\idch)$\footnote{When the input and output dimensions of the channels are different, unitaries and auxiliary systems are used to make them compatible.}
and we want to upper bound the total QFI of the sequential scheme
\begin{equation}
\label{eq:seqCQFI}
\CQFI_{\vec{\w}}^{N} \coloneqq \max_{\rho, \{ \mathcal{V}_i \}}\sum_{x=1}^{\npar} \w_x \QFI \left(\mathcal{E}_{\theta_x}^N (\rho) \right).
\end{equation}
We stress that the optimal strategy, i.e. not only the initial state but also the control unitaries, cannot use any information about which channel $x$ is applied $N$ times.
The main result for the general strategy is the following.
\begin{theorem}
\label{theo:sequentialCE}
The total channel QFI for an adaptive strategy with $N$ uses of the channel satisfies the bound
\begin{widetext}
\begin{equation}
\label{eq:seqCEbound}
\CQFI_{\vec{\w}}^N\leq 4 \min_{\mathfrak{h}} \left\{ N \left\lVert \sum_{x}\w_x \alpha_x  \right\rVert +  N ( N - 1) \max_x (\norm{ \beta_x }) \left[  \norm{ \sum_{x=1}^\npar \w_x \beta_x }  + 2 \sqrt{  \left( \sum_{x=1}^{\npar} \w_x \right) \norm{ \sum_{x=1}^\npar \w_x \alpha_x } } \right] \right\}.
\end{equation}
\end{widetext}
where $\alpha_x = \left( \partial_{\theta_x}\!\vec{K}_x - \I h_x \vec{K}_x \right)^\dag \left( \partial_{\theta_x}\!\vec{K}_x - \I h_x \vec{K}_x \right)$ and $\beta_x = \left( \partial_{\theta_x}\!\vec{K}_x - \I h_x \vec{K}_x \right)^\dag \vec{K}_x$.
\end{theorem}
The proof is relegated to Appendix~\ref{app:proofTheoSeq}.
For $\npar=1$ one obtains the known single-parameter bound~\cite{Demkowicz-Dobrzanski2014,Sekatski2016,Zhou2020}.

The parallel strategy, corresponding to the channel $\mathcal{E}_{\theta_x}^{\otimes N}$, is less powerful than a sequential one since it can be obtained by choosing swap operations as the control unitaries.
For completeness, we provide a complete derivation of the tighter bound for the parallel strategy in Appendix~\ref{app:parallel}.
Note that the two bounds differ asymptotically only when Heisenberg scaling is allowed.
However, since we are mostly interested in noisy channels that satisfy the conditions $\beta_x=0 \; \forall \, x=1,\dots,\npar$  this distinction will not be relevant.
In this case, the optimal upper bound in~\eqref{eq:seqCEbound} is asymptotically linear in $N$:
\begin{gather}
\label{eq:SQLbound}
\CQFI_{\vec{\w}}^N \leq N \CQFIbnd_{\vec{\w}}, \quad
\CQFIbnd_{\vec{\w}} \coloneqq 4 \min_{\mathfrak{h},\beta_x = 0} \norm{ \sum_{x=1}^{\npar} \w_x \alpha_x } \;
\end{gather}
We call $\CQFIbnd_{\vec{\w}}$ the asymptotic SQL bound, since in the limit $N \gg 1$ the optimal variables $h_x$ in~\eqref{eq:seqCEbound} must make the quadratic term vanish.

Furthermore, in Appendix~\ref{app:MarkovianBound}, we provide a time-continuous variant of the bound that applies to general Markovian noise, where the duration of a single probing step may be adjusted arbitrarily.
In this case it is the total interrogation time that is treated as a resource and the bound is a direct generalization of the single-parameter bounds derived in~\cite{Demkowicz-Dobrzanski2017,Zhou2017}.

The condition $\beta_x=0$ has been dubbed ``Hamiltonian in the Kraus span'' (HKS) condition:
\begin{equation}
\begin{split}
\beta_x = 0
\iff
\I \partial_{\theta_x}\!\vec{K}_x^\dag \, \vec{K}_x \in \mathrm{span}_{\mathbb{R}}( K_{x,i}^{\dag} K_{x,j}, \forall i,j  ),
\end{split}
\end{equation}
where we have introduced the so-called Kraus spans~\cite{Zhou2020} of the channels $\mathcal{E}_{\theta_x}$.
Here $\I \partial_{\theta_x}\!\vec{K}_x^\dag \, \vec{K}_x$ is not a real Hamiltonian, but an effective generator for the parameter $\theta_x$.
In the simple case of parameter-independent noise following a single-parameter unitary we have the Kraus operators $K_j e^{-\I \theta H}$ and indeed $H = \I \partial_{\theta} \vec{K}^\dag \vec{K}$.
For an arbitrary parameter of a full-rank channel this condition is always satisfied, thus ruling out Heisenberg scaling $N^2$ for almost all (in a measure-theoretical sense) quantum channels~\cite{Fujiwara2008}.

Just like in the single-use case, the sum of single-parameter bounds is an upper bound that does not take into account probe incompatibility and the triangle inequality implies $\CQFIbnd_{\vec{\w}} \leq \sum_{x=1}^\npar \w_x \mathfrak{B}_{\theta_x}$, where $\mathfrak{B}_{\theta_x} \coloneqq 4 \min_{h_x} \norm{ \alpha_x }$ subject to $\beta_x = 0$.
Following this observation, we introduce an asymptotic probe incompatibility measure that generalizes~\eqref{eq:probe_inc_cost}.
In the SQL case, it is obtained by evaluating $\CQFIbnd_{\vec{\w}}$ with weights $\w_x = 1/\mathfrak{B}_{\theta_x}$:
\begin{gather}\mathfrak{I}_{\infty}(\mathcal{E}_{\parvec})\coloneqq \npar \Biggl( 4 \min_{\substack{ \mathfrak{h}\\ \{ \beta_x=0\} } } \norm{ \sum_{x=1}^{\npar} \frac{\alpha_x}{\mathfrak{B}_{\theta_x}}  }\Biggr)^{\!\!-1}.
\end{gather}
Since the single parameter bounds $\mathfrak{B}_{\theta_x}$ are asymptotically attainable~\cite{Zhou2020}, the quantity $\mathfrak{I}_{\infty}$ is a computable lower bound on the actual asymptotic incompatibility.

\subsection{Purification-based definition of the QFI matrix}

In the proof of Theorem~\ref{theo:multiparFujiImaiGen} we have used the established purification-based definition of the single-parameter QFI.
Since we have worked assuming that each channel $\mathcal{E}_{\theta_x}$ can be different, the purifications pertaining to different parameters need not be related.
However, in the case of multiple parameters and a single channel, the need to choose a different purification for each parameter would entail that the bound is not tight. In other words, the upper bound~\eqref{eq:singleuseCEbound} would not necessarily correspond to the weighted trace of the QFI matrix of the optimal probe state on which the channel acted upon.

In this section we show that this is not the case and that the purification-based definition of the scalar QFI is easily generalized to the matrix valued case\footnote{A similar statement appeared in~\cite{Yue2014} without explicit proof.}, meaning that each choice of the matrices $\mathfrak{h}$ corresponds to a unique purification.

For notational convenience we introduce the $d{\times}\npar$ Jacobian matrix of a $d$-dimensional pure state $\ket{\Psi_{\parvec}}$: $\boldsymbol{\nabla} \Psi_{\parvec} = [  \ket{ \partial_1 \Psi_{\parvec}} \dots \ket{\partial_\npar  \Psi_{\parvec}} ]$ so that we can write the $\npar{\times}\npar$ QFI matrix compactly as
\begin{equation}
\QFI( \ket{\Psi_{\parvec}}) = 4  \Re\left[(\boldsymbol{\nabla} \Psi_{\parvec})^\dag \boldsymbol{\nabla} \Psi_{\parvec} - (\boldsymbol{\nabla} \Psi_{\parvec})^\dag | \Psi_{\parvec} \rangle \langle \Psi_{\parvec} | \boldsymbol{\nabla} \Psi_{\parvec} \right],
\end{equation}
and clearly $\QFI( \ket{\Psi_{\parvec}}) \leq 4 \Re\left[(\boldsymbol{\nabla} \Psi_{\parvec})^\dag \boldsymbol{\nabla} \Psi_{\parvec} \right]$ since the term subtracted is a positive-semidefinite matrix.
We can thus formulate the result as follows.
\begin{theorem}
\label{th:purifdefQFIM}
The QFI matrix of a mixed state $\rho_{\parvec}$ is equal to the minimal (in the positive-semidefinite sense) QFI matrix of its purifications
\begin{equation}
\label{eq:purifdefQFIM}
\QFI( \rho_{\parvec} ) = \min_{ \Psi_{\parvec} } \QFI(\ket{\Psi_{\parvec}}) = 4 \min_{\Psi_{\parvec} } \Re\left[(\boldsymbol{\nabla} \Psi_{\parvec})^\dag \boldsymbol{\nabla} \Psi_{\parvec} \right].
\end{equation}
\end{theorem}
\begin{proof}
Starting from an arbitrary fixed purification $\ket{\Psi_{\parvec}}$, such that $\rho_{\parvec}=\Tr_E | \Psi_{\parvec} \rangle \langle \Psi_{\parvec} |$, all other purifications of the quantum statistical model are obtained by acting with a parameter-dependent unitary on the environment $\ket{\tilde{\Psi}_{\parvec}}=\id \otimes u_{\parvec} \ket{\Psi_{\parvec}}$.
If we consider the matrix $(\boldsymbol{\nabla} \tilde{\Psi}_{\parvec})^\dag \boldsymbol{\nabla} \tilde{\Psi}_{\parvec} $ we see that the unitary $u_{\parvec}$ only enters through the quantities $u_{\parvec}^\dag \partial_{j} u_{\parvec} = - \I h_j$, where $h_j$ are Hermitian matrices.
Crucially, the matrices $h_j$ are independent variables\footnote{
To see explicitly that this is possible, we can choose without loss of generality $u_{\parvec}=e^{-\I \sum_{j=1}^{\npar} (\theta_j - \theta^*_j) h_j}$, where $\parvec^*$ is the true value of the parameters at which all derivatives and functions are evaluated and indeed $u_{\parvec}^\dag \partial_{j} u_{\parvec} \big|_{\parvec =\parvec^*}  = - \I h_j$}.
From the purification-based definition of the single-parameter QFI~\cite{Fujiwara2008,Koodynski2013}
$\QFI_{jj}(\rho_{\parvec})=\min_{h_j} \braket{ \partial_j \Psi_{\parvec} | \partial_j \Psi_{\parvec} } = \QFI_{jj}( \ket{\Psi^{*}_{\parvec}} )$,
we know that for each parameter there exists an optimal matrix $h_j$ such that the pure state model $ \ket{\Psi^{*}_{\parvec}} $ satisfies $\ket{\partial_j \Psi^*_{\parvec}} = \frac{1}{2} L_j \otimes \id \ket{\Psi^*_{\parvec}}$, where $L_j$ is the SLD of the original mixed state~\cite{Fujiwara2008,Kolodynski2014}, also implying $\braket{\partial_j \Psi^*_{\parvec}| \Psi^*_{\parvec}} = 0$.
Therefore for this purification we obtain the equality
\begin{equation}
\label{eq:complexQFIequalityPurif}
\Tr [ \rho_{\parvec} L_i L_j ] = 4 \braket{ \partial_i \Psi^*_{\parvec} | \partial_j \Psi^*_{\parvec} } \implies \QFI(\rho_{\parvec})=\QFI(\ket{\tilde{\Psi}_{\parvec}}),
\end{equation}
where the QFI matrix is the real part of the complex matrix on the left.
Finally, thanks to the monotonicity of the QFI matrix~\cite{Petz1996a,Liu2019d}, we have the matrix inequalities
\begin{equation}
\label{eq:purifIneqQFImat}
\QFI(\rho_{\parvec}) \leq \QFI(\ket{\tilde{\Psi}_{\parvec}}) \leq 4 \Re \left[ (\boldsymbol{\nabla} \tilde{\Psi}_{\parvec})^\dag \boldsymbol{\nabla} \tilde{\Psi}_{\parvec} \right],
\end{equation}
which hold for arbitrary purifications $\ket{\tilde{\Psi}_{\parvec}}$ since partial tracing is a CPTP map.
This shows that the matrix-valued minimization in the statement of the theorem is well posed and closes the proof.
\end{proof}
We remark that when $ \Im \left( \Tr [ \rho_{\parvec} L_i L_j ] \right) = 0 \;\; \forall \, i,j$ the bound $\Tr [ W \QFI(\rho_{\parvec})^{-1} ] $ is equal to the Holevo Cramér-Rao bound and thus asymptotically attainable~\cite{Ragy2016,Suzuki2018}.
Therefore, thanks to the equality between the complex matrices~\eqref{eq:complexQFIequalityPurif}, this condition can be checked from the optimal purification.

By solving the minimization~\eqref{eq:singleuseCEbound} one obtains the maximal value of trace of the QFI matrix, but the optimal state is not identified and it is not immediate to retrieve the full QFI matrix.
In Appendix~\ref{app:optistate} we present an algorithm that finds an optimal state by enforcing that a solution of the minimax problem in~\eqref{eq:maxmin} must be a saddle point in the variables $\rho$ and $\mathfrak{h}$, analogously to the single-parameter case~\cite{Zhou2020}.
The same approach can be applied to the asymptotic SQL bound~\eqref{eq:SQLbound}, but the meaning of the obtained $\rho$ is unclear, unlike for single-parameter problems where it is connected to the approximate error strategy to attain the bound~\cite{Zhou2020}.
Once an optimal state is found, it is easy to obtain its full QFI matrix using the purification-based definition.

\section{Applications to multiparameter quantum metrology}
\label{sec:applications}

In this section we apply the theory we have developed to physical problems.
First, we study the paradigmatic problem of characterizing the Hamiltonian of a $d$-dimensional quantum system, i.e. Hamiltonian tomography, in which the number of parameters scales with the dimension of the Hilbert space.
We consider the effect of erasure noise and we show that asymptotic probe incompatibility is identical to the noiseless case.
Focusing on the submodel with only commuting generators the problem is equivalent to multi-phase estimation in presence of optical losses, a practical application for which we confirm and strengthen existing results.
Then we consider the problem of estimating a phase and a noise parameter, for photon loss and phase diffusion.
In these instances we see no asymptotic probe incompatibility, matching physical intuition and previous indications.

An additional example, the estimation of the two noise parameters of a generalized amplitude damping channel, is studied in Appendix~\ref{app:GADchannel}.
For this channel there is no advantage in using advanced strategies, since the asymptotic SQL bound coincides with the single-use one, analogously to several single-parameter qubit noise channels~\cite{Koodynski2013}.

\subsection{Hamiltonian tomography with erasure noise}

Before focusing on our specific model, we start with a more general consideration.
The vector $\parvec$ parametrizes the unitary $U_{\parvec}=e^{-\I \sum_{x}^\npar \theta_x G_x }$, where $G_x$ are Hermitian generators.
We model the noise with a channel that acts \emph{after} the unitary encoding, so that the Kraus operators are $K'_i = K_i U_{\parvec}$, where $K_i$ are the Kraus operator of the noise only, and their derivatives are $\partial_x K'_i = -\I K'_i G_x$, since we assume to work at the true value $\theta_x = 0 \, \forall x$.
Using the HKS conditions, the SQL bound~\eqref{eq:SQLbound} can be simplified under these assumptions:
\begin{equation}
\label{eq:SQLnoiseafter}
\CQFIbnd_{\vec{\w}} = 4 \min_{ \substack{ h_x \\ \beta^x = 0} }  \norm{ - \left( \sum_{x=1}^\npar \w_x G_x^2 \right) + \vec{K}^\dag \left(\sum_{x=1}^\npar \w_x h_x^2 \right) \vec{K}}.
\end{equation}

Hamiltonian tomography of a $d$-dimensional quantum system amounts to the estimation of the $d^2-1$ parameters of a $\mathrm{SU}(d)$ transformation.
This problem, or slight variations, has been studied in the noiseless and error-corrected scenario~\cite{Ballester2004,Kahn2007,Imai2007,Yuan2016b,Kura2017,Gorecki2020}.
It has been shown that an advantage of $O(\sqrt{\npar})=O(d)$ is possible using a simultaneous adaptive estimation strategy.

Here, for simplicity we consider the $\npar=d^2$ parameters of a $\mathrm{U}(d)$ transformation.
This choice makes the QFI matrix singular, but it is not a problem when using the total QFI as a figure of merit.
Concretely, we choose the following generators
\begin{equation}
\label{eq:HamGenerators}
\begin{split}
G_{j}^\t{diag} = | j \rangle \langle j | \qquad
& G_{j k}^{\t{re}} = \frac{1}{2}\left( | j \rangle \langle k | + | k \rangle \langle j | \right) \\
&G_{j k}^{\t{im}} = \frac{\I}{2} \left( | j \rangle \langle k | - | k \rangle \langle j | \right),
\end{split}
\end{equation}
where we have separated the three submodels with $d$ diagonal, $d(d-1)/2$ real and $d(d-1)/2$ imaginary off-diagonal generators.
For the noise we consider a qudit erasure channel, described by the following $d+1$ Kraus operators
\begin{equation}
\label{eq:ErasureQuditKraus}
K_0 = \sqrt{\eta} \begin{bmatrix}
 \id_{d} \\
0  \dots 0
\end{bmatrix} \quad K_i = \sqrt{1-\eta}
| d + 1 \rangle \langle i | 
\end{equation}
for $i=1,\dots,d$; the output Hilbert space contains an additional dimension that represents the system in the ``lost'' state $\ket{d+1}$.

The asymptotic SQL bound on the total QFI is the following (more details in Appendix~\ref{app:HamTomErause})
\begin{gather}
\CQFIbnd = \frac{\eta}{1-\eta}\left( \CQFI_{\mathrm{diag}} + \CQFI_{\mathrm{real}} +  \CQFI_{\mathrm{imag}} \right) \label{eq:HamEstBnd} \\
\CQFI_{\mathrm{diag}} = \frac{4(d-1)}{d},  \quad
 \CQFI_{\mathrm{real}} = \CQFI_{\mathrm{imag}} = d-1,
\end{gather}
where the weights are $\w_i=1$ (in this case we will drop the weight index) and $\CQFI$ are the \emph{noiseless} bounds for a single use of the channel for the three submodels (attained by considering the extended channel and a probe state maximally entangled with the auxiliary system~\cite{Yuan2016b}).
This means that probe incompatibility exists only inside the three submodels.

The single-parameter bounds are identical for all the parameters: $\mathfrak{B}_{x} = \frac{\eta}{1-\eta}$, equivalent to the estimation of a qubit rotation with erasure noise~\cite{Escher2011, Demkowicz-Dobrzanski2012,Demkowicz-Dobrzanski2014}.
This means that the asymptotic probe incompatibility cost is simply obtained by rescaling the bound~\eqref{eq:HamEstBnd} and it is identical to the noiseless case $\mathfrak{I}_{\infty} = \frac{d^3}{2(d^2 +d -2)} = O(\sqrt{\npar})$, which is an intermediate scaling between the two
extremal scalings $O(1)$ (compatibility) and $O(\npar)$ (maximal incompatibility).

\subsection{Lossy multiple-phase estimation}
\label{subsec:lossymultiphase}

Estimating multiple optical phases simulatenously is a paradigmatic task in multiparameter quantum metrology~\cite{Macchiavello2003,Ballester2004a,Humphreys2013,Gagatsos2016a,Zhang2017a,You2017a,Gessner2018,Li2019a,Goldberg2020,Markiewicz2020}, also realized experimentally~\cite{Polino2019,Valeri2020}.
In the noiseless scenario, Humphreys et al.~\cite{Humphreys2013} argued that a simultaneous estimation strategy provides an advantage in the total variance over individual quantum estimation schemes that scales as $O(\npar)$.
However, the apparent advantage is actually a pitfall in the application of the QCRB. A minimax analysis taking into account the \emph{total} number of photons, shows that no advantage in scaling with $\npar$ is present~\cite{Gorecki2021}.
Fortunately, asymptotic discrepancies between the minimax, Bayesian and QCRB  predictions may only appear in the Heisenberg scaling scenarios and disappear when the optimal scaling corresponds to the SQL and this issue will not have any impact on our results.

Here we study the problem in the presence of optical losses that forbid Heisenberg scaling.
This setting was first studied in~\cite{Yue2014} with a purification-based bound and the problem was also studied treating the transmissivity as a nuisance parameter~\cite{Yang2018a}.
These works considered indistinguishable photons and thus a parallel strategy.
Here we work in a more general scenario, treating the photons as distinguishable particles on which adaptive strategies can be applied, e.g. thanks to an additional degree of freedom.

Mathematically, this problem is a commuting submodel of the full $\mathrm{U}(d)$ model we have introduced before.
We consider only $\npar=d-1$ parameters corresponding to all the diagonal generators in Eq.~\eqref{eq:HamGenerators} except one, representing the reference arm of the interferometer with a known phase.
Since the generators commute, the bound~\eqref{eq:SQLnoiseafter} now holds for all values of the parameters.
The Kraus operators~\eqref{eq:ErasureQuditKraus} describe photon losses happening with equal probability $1-\eta$ in each mode.

We obtain the following bound on the total variance
\begin{equation}
\label{eq:multiphbound}
\Delta^2 \tilde{\parvec} \geq \frac{1-\eta}{\eta} \frac{\npar^3}{4 N (\npar-1)} \geq  \frac{1-\eta}{4 \eta} \frac{\npar^2}{N},
\end{equation}
valid for $\npar > 1$ ($d>2$), since the total QFI is identical to the bound~\eqref{eq:HamEstBnd} on diagonal parameters, with the substitution $d \to d-1$, for $d>2$.
As explained more in details in Appendix~\ref{app:lossmultiphaseDeriv}, when the same $\npar$ diagonal generators act on higher-dimensional systems the optimal total QFI is unchanged.
This means that adding additional reference modes does not improve the bound.
The rightmost quantity in~\eqref{eq:multiphbound} corresponds to the bound obtained in~\cite{Yue2014} for simultaneous estimation of the $\npar$ phase with $N$ total indistinguishable photons, in the limit $N \gg 1, \npar \gg 1 , N / \npar \gg 1$.
Our result shows that the same bound holds even for the most general adaptive strategy acting on distinguishable photons.
The asymptotically attainable total variance for a separate estimation strategy with $N$ total photons is $\frac{1-\eta}{\eta} \frac{\npar^2}{N}$, therefore simultaneous estimation could potentially reduce the total error of a factor $4$.
The same factor $4$ advantage is observed in the regime $\npar \gg 1$ when comparing separate and simultaneous estimation strategies with classical light~\cite{Goldberg2020}, which is however not optimal in presence of lossess.
Practical schemes to attain the bound~\eqref{eq:multiphbound} asymptotically (both in $N$ and $\npar$) are currently not known.

It should already be clear that this sub-model shows stronger probe incompatibility (in terms of the number of parameters involved) than the full model.
Indeed, the asymptotic probe incompatibility cost is almost maximal $\mathfrak{I}_{\infty} = \frac{\npar^2}{4(\npar-1)} = O(\npar)$.
We stress again that it is equal to the probe incompatibility cost of the noiseless problem, since the bound on the total QFI is simply a rescaling of the optimal trace of the noiseless QFI matrix.

\subsection{Phase and loss}

Here we consider the problem of simultaneously estimating the phase shift and the loss (absorption) induced by a sample in one arm of a two-arm interferometer (the case of symmetric loss in both arms is simpler and fully solved in~\cite{Ragy2016}).
This problem is a paradigmatic example of a trade-off between the errors on the two parameters~\cite{Crowley2014}.
On the one hand, this problem shows measurement incompatibility, making an analysis in terms of the QFI not complete; for particular probe states and low photon numbers, the fundamental Holevo Cramér-Rao bound is studied in~\cite{Albarelli2019,Conlon2020}.
On the other hand, the problem shows also probe incompatibility: single-mode Fock states are optimal for loss estimation~\cite{Adesso2009}, but phase-insensitive.

Again, we model this metrological problem using a particle description of photons, so that we can describe a single photon as a qubit~\cite{Demkowicz-Dobrzanski2015a}.
The channel experienced by a single photon is thus described by the following Kraus operators
\begin{equation}
K_0 = \begin{bmatrix}
\sqrt{\eta} e^{-\I \varphi } & 0 \\
0 &  1 \\
0 & 0
\end{bmatrix}
  \qquad  K_1 = \begin{bmatrix}
0 & 0 \\
0 &  0 \\
\sqrt{1-\eta} & 0
\end{bmatrix},
\end{equation}
where the extra output dimension accounts for lost photons and we want to estimate the phase $\varphi$ and the transmissivity $\eta$.

As intuitively expected, with a single use of the channel there is indeed probe incompatibility and we obtain
\begin{equation}
	\label{eq:incompPhLoss}
\mathfrak{I} = 2\Biggl(\frac{1-\eta}{\eta +\sqrt{\eta }-\sqrt{2 (1+\sqrt{\eta })}}\Biggr)^{\!\!2},
\end{equation}
a strictly increasing function of $\eta$ in the range $ 0 \leq \eta \leq 1 $ that spans the values $1 \leq \mathfrak{I} \leq \frac{32}{25}$.
Alternatively, probe incompatibility can be observed from the the fact that the total channel QFI with $\w_i=1$ is equal to the optimal QFI about the transmissivity $\CQFI=\CQFI_{\eta} = \frac{1}{\eta(1-\eta)}$.

On the contrary, the asymptotic bound indicates that there is no incompatibility as we have $\mathfrak{I}_{\infty} = 1$ and $\CQFIbnd = \mathfrak{B}_{\varphi} + \mathfrak{B}_{\eta}$.
More details on the calculations are found in Appendix~\ref{app:phaseloss}.

\subsection{Phase and dephasing}
Simultaneous estimation of an optical phase and its phase-diffusion coefficient is another well-known two-parameter problem~\cite{Vidrighin2014,Knysh2013,Szczykulska2017}.
At the single photon level the evolution is described by the following Kraus operators
\begin{equation}
K_0 = \sqrt{\frac{1+\eta}{2}} \begin{bmatrix}
e^{\I \varphi } & 0 \\
0 &  1
\end{bmatrix}
  \;\;\;  K_1 = \sqrt{\frac{1-\eta}{2}} \begin{bmatrix}
e^{\I \varphi } & 0 \\
0 &  -1
\end{bmatrix}.
\end{equation}

It is known that there is no probe incompatibility for $N=1$, while it reappears for $N\geq 2$ and then again decreases with $N$~\cite{Ragy2016} and vanishes asymptotically~\cite{Knysh2013}.
Our bounds \eqref{eq:SQLbound} and \eqref{eq:singleuseCEbound} show that both with a single use of the channel and asymptotically there is no probe incompatibility: $\mathfrak{I}= \mathfrak{I}_{\infty} = 1$.
Details on this calculation are presented in Appendix~\ref{app:phasedephasing}.
This asymptotic disappearance of incompatibility agrees with the results of a direct QFI maximization in parallel strategies, reported in~\cite{Ragy2016,Knysh2013}.

\section{Application to quantum channel discrimination}
\label{sec:discrimination}
In this section we draw a connection between bounds on the total QFI introduced in Sec.~\ref{sec:multiQFIsum} and the problem of discriminating between several channels.
First, we generalize the framework to arbitrary quantum channel discrimination tasks.
This approach is particularly suited to those problem where there is some reference channel to which the channels are naturally related.
In this framework we derive an inequality that we call a speed limit for the discrimination of multiple noisy channels.
As an application, we derive bounds on the performance of a time-continuous version of Grover's algorithm in presence of noise, revisiting the approach of~\cite{Demkowicz-Dobrzanski2015}.
Our new bounds allow us to close a gap in the proof that was left open as a conjecture.

\subsection{Background notions}

\subsubsection{Probability of error}
The error in discriminating among $\npar$ states $\rho_n$ with prior probability $p_x$ is given by~\cite{Watrous2017}
\begin{equation}
\label{eq:generic_error}
P_{\mathrm{H}} \left( \{ \rho_x , p_x \}  \right) \coloneqq 1 -  \max_{\sum_x \Pi_x = \id} \sum_{x=1}^\npar p_x \Tr ( \rho_x \Pi_x ),
\end{equation}
where H stands for ``Helstrom error'' and $\{ \Pi_x \}$ is a $\npar$-outcome POVM.
For $\npar=2$ and equal priors $p_1=p_2 = 1/2$ it reduces to
\begin{equation}
\label{eq:binary_error}
P_H \left(\rho_1,\rho_2\right) = \frac{1}{2} \left[ 1 - D_{\mathrm{tr}}(\rho_1,\rho_2) \right],
\end{equation}
where
\begin{equation}
D_{\mathrm{tr}}(\rho_1,\rho_2) \coloneqq \frac{1}{2}\norm{ \rho_1 - \rho_2 }_1
\end{equation}
is the trace distance and $\norm{ A }_1 = \Tr \sqrt{ A^\dag A } $ is the trace norm, i.e. the sum of the singular values of $A$.
It is useful to introduce the Fuchs-van de Graaf inequalities between the trace distance and the fidelity~\cite{Fuchs1999}
\begin{equation}
\label{eq:FuchsVDGineq}
1 - F(\rho_1,\rho_2) \leq D_{\mathrm{tr}}(\rho_1,\rho_2) \leq \sqrt{1 - F(\rho_1,\rho_2)^2},
\end{equation}
where we define the fidelity as
\begin{equation}
F(\rho_1,\rho_2) \coloneqq \norm{ \sqrt{\rho_1} \sqrt{\rho_2} }_1 = \Tr \sqrt{\sqrt{\rho_1} \rho_2 \sqrt{\rho_1}}.
\end{equation}

There is no closed form solution for the generic multi-hypothesis problem, but it is possible to find upper and lower bounds.
A lower bounds in terms of binary discrimination is the following~\cite{Qiu2008}
\begin{equation}
\label{eq:multiboundfrombinary}
P_H \left( \{ \rho_x , p_x \}  \right) \geq \frac{1}{2}\left( 1 - \frac{1}{\npar-1} \sum_{1 \leq x < y \leq \npar}  \norm{ p_x \rho_x - p_y \rho_y }_1   \right).
\end{equation}

\subsubsection{Bures angle}
The angular Bures distance or Bures angle is defined in terms of the fidelity~\cite{bengtsson2017geometry}
\begin{equation}
\label{eq:BuresAngle}
D_\mathrm{A}(\rho_1,\rho_2) \coloneqq \arccos F(\rho_1,\rho_2)
\end{equation}
and in particular the infinitesimal version is related to the QFI
\begin{equation}
\label{eq:infinitAngle}
D_\mathrm{A}(\rho_{\theta},\rho_{\theta+ \mathrm{d} \theta}) = \frac{1}{2} \sqrt{\QFI(\rho_\theta) } \mathrm{d} \theta,
\end{equation}
where we have introduced a suitable smooth parametrization.
Crucially, $D_\mathrm{A}(\rho_1,\rho_2)$ is the length of a geodesic path between $\rho_1$ and $\rho_2$ with respect to this infinitesimal metric~\cite{Uhlmann1992,Taddei2013}, thus for any smooth parametrization such that $\rho_1=\rho_{\theta=0}$ and $\rho_2=\rho_{\theta=\theta^*}$ we have
\begin{equation}
\label{eq:geodesicIneq}
D_\mathrm{A}(\rho_1,\rho_2) \leq  \frac{1}{2} \int_{0}^{\theta^*} \sqrt{\QFI(\rho_\theta) } \mathrm{d} \theta.
\end{equation}

We use the second inequality~\eqref{eq:FuchsVDGineq} to upper bound the trace distance with the Bures angle
\begin{equation}
\label{eq:boundDtrDA}
D_{\mathrm{tr}}(\rho_1,\rho_2)  \leq \sqrt{1 - \cos^2[D_{\mathrm{A}}(\rho_1,\rho_2)] }\leq  D_{\mathrm{A}}(\rho_1,\rho_2),
\end{equation}
where we have used the inequality $\cos^2 \xi \geq 1-\xi^2 $ for $0\leq \xi \leq \pi / 2$.

\subsection{Reference-state bound for state discrimination}

Thanks to inequality~\eqref{eq:boundDtrDA} and the triangle inequality the trace distance can be upper bounded as
\begin{equation}
D_{\mathrm{tr}}(\rho_x,\rho_y) \leq D_{\mathrm{A}}(\rho_x,\rho_0) + D_{\mathrm{A}}(\rho_y,\rho_0),
\end{equation}
where we have introduce an arbitrary reference state $\rho_0$.
Now we make the crucial assumption that the states $\rho_x$ and $\rho_y$ are parametrized by a parameter $\theta$ such that $\rho_x = \rho_{x,\theta = \theta^*}$, $\rho_y = \rho_{y,\theta= \theta^*}$ and $\rho_0 = \rho_{x,\theta = 0} = \rho_{y,\theta = 0} $, thanks to~\eqref{eq:geodesicIneq} we have
\begin{equation}
\label{eq:DtrIntFIneq}
D_{\mathrm{tr}}(\rho_x,\rho_y) \leq \frac{1}{2} \left[ \int_{0}^{\theta^*} \! \mathrm{d} \theta  \, \sqrt{\QFI(\rho_{x,\theta})} + \sqrt{\QFI(\rho_{y, \theta})}  \right].
\end{equation}
The intuition of using the geodesic and the triangle inequalities is pictorially represented in Fig.~\ref{fig:schemes}-(c).

Now we focus on discriminating $\npar$ states with uniform prior $\w_x=\frac{1}{\npar}$ and apply~\eqref{eq:multiboundfrombinary} to find a lower bound on the error from pairwise discrimination.
Thanks to~\eqref{eq:DtrIntFIneq} the sum appearing in~\eqref{eq:multiboundfrombinary} is upper bounded as follows
\begin{multline}
\label{eq:sumDtrIneqsumQFI}
\sum_{ 1 \leq x < y \leq \npar } D_{\mathrm{tr}}(\rho_x,\rho_y)  \leq \frac{(\npar-1)}{2} \int_{0}^{\theta^*} \! \mathrm{d} \theta \sum_{ x=1 }^\npar \sqrt{\QFI(\rho_{x,\theta})} \\
\leq \frac{(\npar-1)\sqrt{\npar}}{2} \int_{0}^{\theta^*} \! \mathrm{d} \theta \sqrt{\sum_{ x =1 }^\npar \QFI(\rho_{x,\theta})},
\end{multline}
where first we have used the identity $ \sum_{1 \leq x < y \leq \npar} (\w_x + \w_y) =  (\npar-1) \sum_x \w_x $ and then Jensen inequality $\sum_{i=1}^\npar \sqrt{\w_i} \leq \sqrt{\npar} \sqrt{ \sum_{i=1}^\npar \w_i}$.
Using~\eqref{eq:sumDtrIneqsumQFI} in conjunction with~\eqref{eq:multiboundfrombinary} we obtain
\begin{equation}
\label{eq:boundsumQFIepsilon}
\int_{0}^{\theta^*} \! \mathrm{d} \theta \sqrt{ \frac{1}{\npar} \sum_{ x =1 }^\npar \QFI(\rho_{x,\theta})} \geq 1-2 P_H \left( \{ \rho_x , \npar^{-1} \}  \right).
\end{equation}

\subsection{Speed limits for discriminating noisy quantum channels}

The derivations of the previous section have limited utility for discriminating an ensemble of quantum states.
In general, for a fixed set of states the sum $\sum_{x=1}^\npar \QFI( \rho_x ) = O(\npar)$, since there are no issues with probe incompatibility.
Thanks to the additivity of the QFI we thus get a trivial bound (constant in the number $\npar$ of states) on the number of copies $N$ needed to discriminate with fixed accuracy $N \geq O(1) $.
This is a weak and non informative bound, since it is known that the dependence on $\npar$ should be logarithmic~\cite{Harrow2012a}.
However, such bounds reveal their usefulness when applied to channel discrimination, where we can take advantage of general and powerful upper bounds on the total QFI that we have introduced.
In particular, we focus on the case when the HKS conditions are satisfied for all channels, which is always true for full-rank channels.

We consider the discrimination of an ensemble of channels $\{ \mathcal{E}_x \}_{x=1}^{\npar}$ with uniform prior, assuming $N$ repeated uses and a general adaptive strategy with auxiliary systems.
We fix a target error $\varepsilon$ and we find a lower bound on the number $N$ of uses of the channel (i.e. channel queries) needed to go below the error threshold.
The problem amounts to identification of the minimal value of $N$ for which
\begin{equation}
\label{eq:targetmaxPH}
\min_{\rho_0,\{ \mathcal{V}_i \} } P_H \left( \left\{ \mathcal{E}^{N}_x(\rho_0), \npar^{-1} \right\} \right) \leq \varepsilon,
\end{equation}
where $\mathcal{E}^{N}_x(\rho_0)$ is the overall channel of the adaptive strategy, as defined in Sec.~\ref{sec:multiQFIsum}.

Instead of a reference state we now introduce a reference channel $\mathcal{E}_0$ and a suitable parametrization such that $\mathcal{E}^{x} = \mathcal{E}^{x}_{\theta = \theta^*}$ and $\mathcal{E}_0 = \mathcal{E}_{x,\theta = 0}$, as in Fig.~\ref{fig:schemes}-(c).
In this way we can apply~\eqref{eq:boundsumQFIepsilon} and obtain
\begin{equation}
\int_{0}^{\theta^*} \! \mathrm{d} \theta \sqrt{ \frac{1}{\npar} \min_{\rho_0,\{ \mathcal{V}_i \} } \sum_{ x =1 }^\npar \QFI \left( \mathcal{E}^{N}_{x,\theta}(\rho_0) \right)} \geq 1-2\epsilon.
\end{equation}

Here we see that the bounds on the total QFI $\CQFI^N$ that we have derived in Sec.~\ref{sec:multiQFIsum} are useful even when considering different channels.
In particular, if we can find a parametrization such that the HKS conditions are satisfied for all the channels we obtain the following inequality
\begin{equation}
\label{eq:sequentialMbound}
N \geq \frac{ \npar (1- 2\varepsilon)^2 }{\left( \int_0^{\theta^*} \mathop{\mathrm{d} \theta} \sqrt{ \CQFIbnd(\theta)} \right)^2 },
\end{equation}
where we have highlighted that the SQL bound $\CQFIbnd(\theta)$ generally depends on the value $\theta$ at which it is evaluated.
We call this inequality a speed limit, a term which is more appropriate when taking a continuous limit, as done in Appendix~\ref{app:MarkovianBound} for Markovian noise.

We highlight that the speed limit~\eqref{eq:sequentialMbound} gives a nontrivial result in the regime where $\varepsilon$ is small, i.e. in the regime of large distinguishability.
On the contrary, for $\varepsilon = \frac{1}{2}$ the bound is completely trivial, yielding $N\geq 0$.
We also remark that in general it is not possible to perfectly discriminate between  noisy quantum channels with a finite number of uses~\cite{Duan2009} (or finite time for Markovian noise~\cite{Chen2019f}), unlike for unitary channels~\cite{Acin2001}.

If instead of fixing a target error~\eqref{eq:targetmaxPH} we require an identical target fidelity between all pairs of final states, which, in terms of the Bures angle, means
\begin{equation}
\label{eq:targetDA}
D_{\mathrm{A}} \left[ \mathcal{E}^{N}_{x}(\rho_0),  \mathcal{E}^{N}_y(\rho_0)  \right] = \delta \quad \forall x \neq y,
\end{equation}
we can derive a slightly tighter bound, similarly to~\cite{Demkowicz-Dobrzanski2015}
\begin{equation}
\label{eq:sequentialTboundBures}
N \geq \frac{ \npar \delta^2 }{\left( \int_0^{\theta^*} \mathop{\mathrm{d} \theta} \sqrt{ \CQFIbnd(\theta)} \right)^2 }.
\end{equation}

\subsection{Loss of computational speedup in presence of noise}
Following~\cite{Demkowicz-Dobrzanski2015} we consider the time-continuous version of Grover's algorithm proposed by Farhi and Gutmann~\cite{Farhi1998}.
The problem of finding a marked element $x$ of a database of size $d$ is recast as as determining which oracle Hamiltonian $G_x = |x\rangle \langle x |$ generates the evolution $U^{x}_{\omega,\tau} = e^{-\I \omega \tau G_x}$, where $\tau$ and $\omega$ are assumed to be known.
When the oracle is noiseless it is possible to find the marked element with a total runtime $O(\sqrt{d})$ instead of the classical $O(d)$, by appropriately driving the system.

It has been shown that a noisy oracle usually destroys the quantum advantage, although this behaviour was only shown for specific noise models~\cite{Regev2012,Temme2014}.
We study the same problem and we confirm the same loss of advantage for a couple of new noise models, highlighting the
generality of the methodology employed.

In the following, we focus on noisy channels that commute with the unitary evolution so that we can choose $\mathcal{E}_{x,\omega,\tau}=\mathcal{E}_{\tau} \circ \mathcal{U}_{x,\omega,\tau}$ where $\mathcal{U}(\cdot)$ denotes the map $U \cdot U^\dag$.
Moreover we focus on Markovian noise with a semigroup property $\mathcal{E}_{\tau+\tau'}=\mathcal{E}_{\tau} \circ \mathcal{E}_{\tau'}$.
Using the speed limit we have introduced, we obtain a bound on the total runtime $T$ under the most general entanglement-enhanced adaptive strategy.
Thanks to the Markovianity assumption the strongest bound is obtained in the limit $\tau\to0$, i.e. when the control operations are applied as frequently as possible.
In this case, when invoking the metrological bound, we are in fact using a direct generalization of the derived bounds adapted to the estimation of multiple Hamiltonian parameters with arbitrary Markovian noise under the most general adaptive strategy, following the same approach used in~\cite{Sekatski2016,Demkowicz-Dobrzanski2017,Zhou2017}, see Appendix~\ref{app:MarkovianBound} for details.

\subsubsection{Uniform qudit dephasing}

We consider the following uniform dephasing channel
\begin{equation}\label{eq:quditDeph}
\mathcal{E}(\rho) = \eta \rho + (1-\eta) \Delta( \rho ),
\end{equation}
where $\Delta( \rho )$ is the completely dephasing channel wrt to the basis of the oracle Hamiltonians, i.e. $\Delta( \rho ) = \mathrm{diag}(\rho)$.
We consider the total QFI of the ensemble $\{ \mathcal{E} \mathcal{U}_{x,\theta} \}_{x=1}^d$,  where the unitaries $U_{x,\theta}=e^{-\I \theta G_x}$. The corresponding SQL bound~\eqref{eq:SQLbound} can be evaluated analytically
\begin{equation}
\CQFIbnd_{\theta}(\eta) = \frac{4 \eta}{1-\eta} \frac{d-1}{d+ \frac{2}{\eta}} \xrightarrow[d \to \infty]{}\frac{4 \eta}{1-\eta};
\end{equation}
the detailed calculation is relegated to Appendix~\ref{app:quditdephasing}.

To apply the above result to the problem of identifying the oracle Hamiltonian we change parametrization to $\omega = \theta / \tau$ and we choose $\eta_{\gamma,\tau}=e^{-\gamma \tau}$ to obtain the Markovian dephasing channel.
For a total runtime $T=N\tau$ the best bound for frequency estimation is thus
\begin{equation}
\label{eq:freq_dephasing_bound}
\CQFIbnd_{\omega} = \lim_{\tau \to 0} \tau \CQFIbnd(\eta_{\gamma,\tau}) = \frac{4 (d-1)}{\gamma  (d+2)} \leq \frac{4}{\gamma}.
\end{equation}
The bound $\CQFIbnd_{\omega}$ can alternatively be obtained directly from the Markovian bound presented in Appendix~\ref{app:MarkovianBound}, using the Lindblad operators $\{L_i = \sqrt{\gamma} |i\rangle\langle i|\}_{i=1}^d$.
Using inequality~\eqref{eq:freq_dephasing_bound} together with~\eqref{eq:sequentialTboundBures} we obtain a lower bound on the total runtime $T$ required to reach the desired Bures angle $\delta$ between all the possible final states:
\begin{equation}
\label{eq:TlowerboundGrover}
T \geq d \frac{ \gamma \delta^2 }{4 \omega^2 }.
\end{equation}
In particular, a successful identification of all the possible generators $G_x$ means that all final states must be perfectly distinguishable, thus $\delta= \frac{\pi}{2}$.
The corresponding bound on the runtime is then $T \geq d \frac{ \gamma \pi^2 }{ 16 \omega^2 }$.
This closes the conjecture of~\cite{Demkowicz-Dobrzanski2015}.

\subsubsection{Erasure noise}

We notice that the SQL bound~\eqref{eq:HamEstBnd} for diagonal generators can be applied to our problem of discriminating the oracle Hamiltonians in presence of Markovian erasure noise.
We just need to substitute $\eta_{\gamma,\tau}=e^{-\gamma \tau}$ as the noise parameter to obtain
\begin{equation}
\CQFIbnd_{\omega} =  \lim_{\tau \to 0} \tau \CQFIbnd_{\theta}(\eta_{\gamma,\tau}) = \frac{4 (d-1)}{d \gamma} \leq \frac{4}{\gamma}.
\end{equation}
Finally from this we get the same bound on the runtime~\eqref{eq:TlowerboundGrover} as for qudit dephasing.

\section{Discussion and Conclusions}

As multiparameter quantum metrology moves to practical applications it becomes crucial to understand its advantages and limitations.
Our first contribution in this regard is conceptual: we have defined a new measure of metrological incompatibility, that only depends on the local geometry of the parametric family of quantum channels.
In this paper we have evaluated lower bounds on this measure by choosing a particular parametrization and considering only probe incompatibility, which indeed represents the main challenge.
However, further efforts to evaluate this incompatibility measure are certainly needed, since it is a quantity that captures all the obstructions imposed both by quantum mechanics and classical statistics on the task of accurately learning multiple properties of a single quantum channel at once.
From a different point of view, quantum incompatibility is understood as the impossibility to jointly implement two or more input-output devices as components of a larger device~\cite{Heinosaari2016}, leading to a notion of incompatible quantum channels~\cite{Heinosaari2017a,Girard2020}.
Understanding possible connections between these diverse notions of incompatibility is a challenging and interesting open question.

As our main technical result, we have derived widely applicable and computable multiparameter metrological bounds, that apply to several different scenarios.
On the one hand, this is immediately useful to filter out spurious proposals and to assess the performance of feasible protocols.
On the other, taking into account probe incompatibility is a fundamental first step towards a full asymptotic theory of noisy multiparameter quantum metrology.
However, to successfully extend the single-parameter analysis~\cite{Zhou2020}, it will be crucial to investigate the attainability of the bounds and to devise optimal asymptotic strategies.
This is the next step needed to uncover the full potential of multiparameter quantum metrology in realistic scenarios.

Another important extension towards realistic situations is to go beyond models where noise acts independently.
In this regard, the optimal precision for single-parameter metrology with non-Markovian noise has been studied using the same purification arguments we have used and generalized, obtaining a way to numerically evaluate it with an SDP~\cite{Altherr2021}.
While work on this topic is overall still in early stages, we foresee that our extension to multiple parameters will be fundamental in this scenario as well.

Speaking of realistic precision bounds, a further word of caution is in order.
When Heisenberg scaling is possible, a careful analysis of truly optimal protocols generally does not agree with a naive local approach, see e.g.~\cite{Hall2012,Jarzyna2015a,Gorecki2019a}, such as the ones we have applied in this paper.
This discrepancy becomes even more critical for multiparameter problems, where it might affect the scaling of the precision with the number of parameters~\cite{Gorecki2021}.
Luckily, in the noisy SQL scenario there is no such problem~\cite{Jarzyna2015a} and therefore the analysis we have presented in this paper holds even for truly optimal protocols.

On a more practical level, in the the examples we have considered our bounds show that the addition of noise on top of a unitary parameter encoding, while prohibiting Heisenberg scaling, does not affect the amount of probe incompatibility.
This also means that the original advantage (or lack thereof) of simultaneous estimation strategies over separable ones appears to be preserved.
This is certainly an intriguing observation that ought to be investigated more in depth in the future.
In particular, proving a practical advantage of simultaneous strategies even under inevitably noisy working conditions could have important consequences for future technological applications.

Finally, as we have shown by studying the problem of noisy Grover search, the applicability of these theoretical tools extends even beyond quantum metrology.
Indeed, tools from quantum estimation theory are now routinely used in quantum thermodynamics~\cite{Guarnieri2019,Hasegawa2019}, quantum speed limits~\cite{Deffner2017}, quantum algorithms and quantum machine learning~\cite{Meyer2021}.
In particular, single-parameter fundamental metrological bounds have recently found an application in the theory of quantum error correction~\cite{Kubica2020,Zhou2020a,Yang2020b}, providing new perspectives and results.
Accordingly, we expect that such a fertile interplay of different research fields will become even more relevant in the future.
We hope that our results will play a role in these endeavours.

\begin{acknowledgments}
We thank Wojtek~G{\'o}recki for fruitful discussions.
We also thank Janek~Ko{\l}ody{\'n}ski and Christian~Gogolin for sharing with us some of their notes and thoughts on proving the conjecture in~\cite{Demkowicz-Dobrzanski2015}.
This work was supported by the National Science Center (Poland) grant No.\ 2016/22/E/ST2/00559.
\end{acknowledgments}

\appendix

\section{Derivation of channel incompatibility bound~\eqref{eq:probe_inc_cost}}
\label{app:parametrizationbound}
The incompatibility measure~\eqref{eq:incompmeasure} can be lower bounded using the QCRB
\begin{equation}
\label{eq:incompQCRB}
\mathfrak{I}^*(\mathcal{E}_{\parvec}) \geq \max_{\{\vec{w}_x\}} \left(\frac{\min
\limits_{\rho}\tr [W \QFI^{-1}(\mathcal{E}_{\parvec}(\rho))]}{\sum\limits_x  \vec{w}_x^T \QFI^{-1}(\mathcal{E}_{\parvec}(\rho_x^*)) \vec{w}_x} \right),
\end{equation}
where we denote by $\rho^*_x = \argmin_{\rho_x} \vec{w}_x^T \QFI^{-1}(\mathcal{E}_{\parvec}(\rho_x)) \vec{w}_x$ the optimal input probe for estimating the $x$-th parameter ($x$-th scalar function).
One could also obtain a tighter lower bound by using the Holevo Cramér-Rao bound~\cite{Holevo2011b,Demkowicz-Dobrzanski2020} on the numerator of~\eqref{eq:incompmeasure} to retain the (asymptotic) effect of measurement incompatibility.

By rescaling the cost vectors
\begin{equation}
\vec{w}_x \rightarrow \frac{\vec{w}_x}{\sqrt{\vec{w}_x^T \QFI^{-1}(\mathcal{E}(\rho^*_x)) \vec{w}_x}},
\end{equation}
we recast the rhs of~\eqref{eq:incompQCRB} as
\begin{equation}
\mathfrak{I}^*(\mathcal{E}_{\parvec}) \geq
 \frac{1}{\npar}\max_{\{\vec{w}_x\}}
  \min\limits_{\rho} \sum_x \frac{ \vec{w}_x^T \QFI^{-1}(\mathcal{E}_{\parvec}(\rho))\vec{w}_x}{ \vec{w}_x^T \QFI^{-1}(\mathcal{E}_{\parvec}(\rho_x^*)) \vec{w}_x}.
\end{equation}
We now lower bound this quantity further, by restricting the class of cost vectors $\vec{w}_x$ to belong to the set of eigenvectors
 of $\QFI(\mathcal{E}(\rho_x^*))$:
\begin{multline}
\mathfrak{I}^*(\mathcal{E}_{\parvec}) \geq
 \frac{1}{\npar}\max_{\{\vec{w}_x \in \mathcal{W}_x\}}\min\limits_{\rho} \sum_x \frac{ \vec{w}_x^T \QFI^{-1}(\mathcal{E}_{\parvec}(\rho))\vec{w}_x}{ \vec{w}_x^T \QFI^{-1}(\mathcal{E}_{\parvec}(\rho_x^*)) \vec{w}_x},
   \end{multline}
where by $\mathcal{W}_x$  we denoted the set of eigenvectors of $\QFI(\mathcal{E}(\rho_x^*))$.
The restriction to eigenvectors of $\QFI(\mathcal{E}(\rho_x^*))$ is a natural assumption, since when minimizing the cost for a given single parameter the optimal state will tend to carry as much information about the desired parameter with as little correlations as possible with the others (as potential correlations increase the diagonal entries of the inverse QFI matrix).
Moreover, when this is not possible for a given parameter $x$, this indicates that the choice of parametrization is not `natural' for the model considered.
We now use the  eigenvector property to take the vectors under the inverse operation and get:
\begin{multline}
\mathfrak{I}^*(\mathcal{E}_{\parvec}) \geq \\
 \frac{1}{\npar}\max_{\{\vec{w}_x \in \mathcal{W}_x\}}
  \min\limits_{\rho} \sum_x \frac{ \vec{w}_x^T \QFI^{-1}(\mathcal{E}_{\parvec}(\rho))\vec{w}_x}{[\vec{w}_x^T \QFI(\mathcal{E}_{\parvec}(\rho_x^*)) \vec{w}_x]^{-1} (\vec{w}_x^T\vec{w}_x)^2},
   \end{multline}

We can now apply a series of inequalities, analogously to Eq.~\eqref{eq:JinvBound}, and arrive at:
\begin{multline}
\mathfrak{I}^*(\mathcal{E}_{\parvec}) \geq
\frac{1}{\npar}\max_{\{\vec{w}_x \in \mathcal{W}_x\}}
  \min\limits_{\rho} \sum_x \frac{ [\vec{w}_x^T \QFI(\mathcal{E}_{\parvec}(\rho))\vec{w}_x]^{-1}}{[\vec{w}_x^T \QFI(\mathcal{E}_{\parvec}(\rho_x^*)) \vec{w}_x]^{-1}} \geq \\
  \npar \max_{\{\vec{w}_x \in \mathcal{W}_x\}}
  \min\limits_{\rho}\frac{1}{ \sum_x [\vec{w}_x^T \QFI(\mathcal{E}_{\parvec}(\rho_x^*)) \vec{w}_x]^{-1} (\vec{w}_x^T F(\rho) \vec{w}_x)}.
\end{multline}
Finally, when we fix $\vec{w}_x$, i.e. choose a particular parametrization, we arrive at Eq.~\eqref{eq:probe_inc_cost}.

\section{Bound for the adaptive strategy (proof of Theorem~\ref{theo:sequentialCE})}
\label{app:proofTheoSeq}
To derive the bound, we need the following inequality for the operator norm.
\begin{lemmaa} \label{lem:normineq}
Given a set of $\npar$ square matrices $\{ A^x \}_{x=1}^\npar$ and $\npar$ sets of rectangular matrices $\{ L_{k_x}^x \}_{k_x=1}^{n_x}$ (with compatible dimensions) we have the following inequality
\begin{equation}
\label{eq:normineq}
\norm{ \sum_{x=1}^\npar \sum_{k_x=1}^{n_x} L_{k_x}^{x \dag} A^x L_{k_x}^{x} } \leq \max_x \left\{ \norm{ A^x } \right\} \norm{ \sum_{x=1}^\npar \sum_{k_x=1}^{n_x} L_{k_x}^{x \dag} L_{k_x}^{x} } .
\end{equation}
\end{lemmaa}
\begin{proof}
Let us define two matrices
\begin{gather}
\tilde{L} = \left[ L_1^1, \dots, L_{n_1}^1, L_1^2, \dots, L_{n_2}^2, \dots, L_1^p, \dots, L_{n_p}^p \right] \\
\tilde{A}=  \bigoplus_{i=1}^\npar  \id_{n_i} \otimes  A_i,
\end{gather}
i.e. $\tilde{A}$ is a block diagonal matrix with each $A_i$ is repeated in $n_i$ consecutive diagonal blocks.
We have that
\begin{equation}\tilde{L}^\dag \tilde{A} \tilde{L} = \sum_{x=1}^\npar \sum_{j=1}^{n_x} L_j^{x\dag} A^x L_j^{x}.
\end{equation}
Thanks to the submultiplicativity of the operator norm we get
\begin{equation}
\norm{ \sum_{x=1}^\npar \sum_{j=1}^{n_x} L_j^{x\dag} A^x L_j^{x} } = \norm{ \tilde{L}^\dag \tilde{A} \tilde{L} } \leq  \norm{ \tilde{L}^\dag } \norm{  \tilde{A}} \norm{ \tilde{L} },
\end{equation}
upon noticing that
\begin{gather}
\norm{ \tilde{A} } = \max_x \left\{ \norm{ A^x } \right\},  \\
 \norm{ \tilde{L}^\dag }^2 =  \norm{ \tilde{L} }^2 = \norm{ \tilde{L}^\dag \tilde{L} } = \norm{ \sum_{x=1}^\npar \sum_{j=1}^{n_x} L_j^{x\dag} L_j^{x} },
\end{gather}
we get the desired result~\eqref{eq:normineq}.
\end{proof}
When reduced to the single-parameter case we get:
\begin{equation}
\label{eq:normineqDDM}
\norm{ \sum_{j=1}^n L_k^\dag A L_k } \leq \norm{ A } \norm{ \sum_{j=1}^n L_k^{\dag} L_k }.
\end{equation}
With this technical tool we can prove the main theorem.

\begin{proof}
First, we notice from the start that the final bound is expressed in terms of operator norms of the functions $\sum_x \alpha_x$, $\beta_x$ and $\sum_x \beta_x$, which only depend on the Kraus operator of a single copy of the channels $\mathcal{E}_{\theta_x}$.
Crucially, these quantities are unchanged by additional unitaries applied between each use of the channel and by trivial extensions to auxiliary Hilbert spaces.
Therefore we derive the bound considering Kraus operators representing a direct $N$-fold sequential application of the channel: $(\mathcal{E}_{\theta_x})^N$, but the results will be valid also in the general case described in the statement including control operations $\mathcal{V}_i$.

A natural choice of Kraus operators for $N$ sequential applications of the channels is $K_{\vec{k}_x}= K_{k_N} \dots K_{k_1}$.
This choice is not necessarily the optimal, but we use it to find an upper bound to the optimal channel bound~\eqref{eq:seqCQFI} as follows
\begin{equation}
\begin{split}
 & 4 \min_{K_{\vec{k}_x}} \left\Vert \sum_{x=1}^\npar \w_x \sum_{\vec{k}_x} \partial_{\theta_x} \! K_{\vec{k}_x}^\dag \, \partial_{\theta_x}\! K_{\vec{k}} \right\Vert  \\
= & 4 \min_{K_{\vec{k}_x}} \Biggl\Vert \sum_{x} \w_x \sum_{\vec{k}_x} \sum_{i,j=1}^{N} \\
 & \;\; K_{k_1}^\dag \dots \partial_{\theta_x}\! K_{k_i}^\dag \dots K_{k_N}^\dag K_{k_N} \dots \partial_{\theta_x} \! K_{k_j} \dots K_{k_1} \Biggr\Vert;
\end{split}
\end{equation}
each multi-index $\vec{k}_x$ depends on $x$, yet for brevity we suppress the dependence on $x$ when expressing its components $k_i \coloneqq k_{x,i}$.
For additional clarity in the final expressions we will reintroduce the dependence on $x$ when the components of the multi-indices disappear from the calculations.
It is understood that the summation over $k_i$ is inside the summation over $x$ and the two cannot be exchanged, since the Kraus operators $K_{k_x}$ depend on $x$.

By splitting the sum over $i,j$ into diagonal and off-diagonal terms and using the triangle inequality and~\eqref{eq:normineqDDM} we get
\begin{widetext}
\begin{equation}
\label{eq:trJboundinterm}
\CQFI^N_{\vec{\w}} \leq 4 \min_{K_{\vec{k}_x}} \left\{ \sum_{i=1}^N \norm{ \sum_x \w_x \sum_{k_x=1}^{r_x} \partial_{\theta_x} \! K_{k_x}^\dag \, \partial_{\theta_x}\! K_{k_x} }  + \sum_{i<j}^{N} \norm{ \sum_{x} \w_x \sum_{k_i, \dots k_j} \partial_{\theta_x} \! K_{k_i}^\dag \, \dots K_{k_j}^\dag \partial_{\theta_x} \! K_{k_j} \dots K_{k_i}  + \text{h.c.} } \right\}.
\end{equation}
\end{widetext}
We introduce the anti-Hermitian operator
\begin{equation}
\I A_i^x = \sum_{k_{i+1}, \dots k_{j}} K_{k_{i+1}}^\dag \dots K_{k_j}^\dag \, \partial_{\theta_x} \! K_{k_j} \dots K_{k_{i+1}},
\end{equation}
so we can rewrite each term inside the second sum over $i<j$ in~\eqref{eq:trJboundinterm} as follows
\begin{widetext}
\begin{equation}
\norm{ \sum_{x} \w_x \sum_{k_i, \dots k_j} \partial_{\theta_x} \! K_{k_i}^\dag \, \dots K_{k_j}^\dag \partial_{\theta_x} \! K_{k_j} \dots K_{k_i}  + \text{h.c.} } =
\norm{ \sum_x \w_x \sum_{k_i=1}^{r_x} ( \sqrt{z_x} \partial_{\theta_x} \! K_{k_i})^\dag \I A_i^x \left( \frac{K_{k_i}}{\sqrt{z_x}} \right)  + \text{h.c.} },
\end{equation}
where we have also introduced an additional set of real positive parameters $z_x>0$ that leaves the quantity unchanged.
We can expand this quantity further
\begin{multline}
\label{eq:ijsumineqCEbound}
\Biggl\lVert \sum_x \w_x  \sum_{k_x=1}^{r_x} \biggr[  \left( \sqrt{z_x} \partial_{\theta_x} \! K_{k_x} + \I \frac{K_{k_x}}{\sqrt{z_x}} \right)^\dag A_i^x \left( \sqrt{z_x} \partial_{\theta_x} \! K_{k_x} + \I \frac{K_{k_x}}{\sqrt{z_x}} \right)  - z_x \partial_{\theta_x} \! K_{k_x}^\dag \, A_i^x \, \partial_{\theta_x} \! K_{k_x}
    -  \frac{1}{z_x} K_{k_x}^\dag A_i^x K_{k_x} \biggr] \Biggr\lVert  \leq \\
\leq 2 \Biggl(  \norm{ \sum_x \w_x \sum_{k_x=1}^{r_x} \partial_{\theta_x} \! K_{k_x}^\dag \, A_i^x K_{k_x} }  + \norm{ \sum_x \w_x \sum_{k=1}^{r_x} \partial_{\theta_x} \! K_k^\dag ( z_x A_i^x ) \partial_{\theta_x} \! K_k }
+  \norm{\sum_x  \w_x \sum_{k_x} K_{k_x}^\dag  \frac{A_i^x}{z_x} K_{k_x} } \Biggr) \leq\\
\leq 2 \Biggl[ \max_{x}( \norm{ A_i^x } ) \norm{ \sum_x \w_x \sum_{k_x}^{r_x} \partial_{\theta_x} \! K_{k_x}^\dag \,  K_{k_x} }
+  \max_{x} ( z_x \norm{ A_i^x } ) \norm{  \sum_x \w_x \sum_k \partial_{\theta_x} \! K_k^\dag \, \partial_{\theta_x} \! K_k }
+ \max_{x} \left( \frac{ \norm{A_i^x}}{z_x}  \right) \left( \sum_x \w_x \right) \Biggr],
\end{multline}
where we have used the triangle inequality first and then inequality \eqref{eq:normineq} to obtain the last line.
From Eq.~\eqref{eq:normineqDDM} we know that \( \norm{A_i^x} \leq \norm{\beta_x} \), where $\beta_x = \sum_{k} \partial_{\theta_x} \! K_k^\dag \,  K_k = \partial_{\theta_x} \! \vec{K}^\dag \, \vec{K} $; using this inequality, setting $z_x=z$ and performing the sums in~\eqref{eq:trJboundinterm} we obtain
\begin{equation}
\label{eq:optCQFIz}
\CQFI^N_{\vec{\w}}\leq 4 \min_{\mathfrak{h},z} \biggl\{ N \norm{ \sum_{x=1}^{\npar} \w_x \alpha_x } +
N ( N - 1) \max_x ( \norm{\beta_x}) \biggl[ \frac{1}{z} \left( \sum_{x=1}^{\npar} \w_x \right)
+  z \norm{ \sum_{x=1}^\npar \w_x \alpha_x }  +  \norm{ \sum_{x=1}^\npar \w_x \beta_x }  \biggr] \biggr\},
\end{equation}
\end{widetext}
and by performing the explicit minimization over the parameter $z$ we obtain the desired result~\eqref{eq:seqCEbound}.
\end{proof}
We notice that it might be possible to obtain a tighter bound by optimizing over the whole set $\{ z_x\}$ that appears in~\eqref{eq:ijsumineqCEbound} instead of fixing all them to be equal to obtain~\eqref{eq:seqCEbound}.
However, the main use of the parameter $z$ in~\eqref{eq:optCQFIz} is to generalize the bound to an infinitesimal timestep, as we do in Appendix~\ref{app:MarkovianBound}.

\section{Bound for the parallel strategy}
\label{app:parallel}

In the parallel strategy, the action of $N$ channels is described as $\{ \mathcal{E}^{\otimes N}_{\theta_x} \}$.
The total QFI for the parallel strategy with $N$ uses is thus
\begin{equation}
\CQFI^N_{\vec{\w}} \coloneqq \max_{\rho \in \rhospace{\hilb_S^{\otimes N}}} \sum_{x=1}^\npar \w_x \QFI \left( \mathcal{E}^{\otimes N}_{\theta_x} (\rho) \right).
\end{equation}

The following result generalizes the analogous single-parameter one~\cite[Th.~5]{Fujiwara2008}.
\begin{theorema}
\label{theo:parallelCE}
The total entanglement-assisted channel QFI for the parallel scheme with $N$ uses is upper bounded as follows
\begin{equation}
\label{eq:parallelCEbound}
\CQFI^N_{\vec{\w}}  \leq 4 \min_{\mathfrak{h}} \left\{ N \norm{ \sum_{x=1}^{\npar} \w_x \alpha_x } +  N ( N - 1) \norm{ \sum_{i=1}^\npar \w_x \beta_x^2 }\right\},
\end{equation}
where $\alpha_x \coloneqq \partial_{\theta_x}\!\tilde{\vec{K}}_x^\dag \partial_{\theta_x}\!\tilde{\vec{K}}_x$ and $\beta_x = \left( \partial_{\theta_x}\!\vec{K}_x - \I h_x \vec{K}_x \right)^\dag \vec{K}_x$.
\end{theorema}
Before proving this theorem we need a simple inequality
\begin{lemmaa}
Given a collection of $\npar$ Hermitian matrices $A_i$ we have the following inequality between operator norms
\begin{equation}
\label{eq:normineqTensor}
\norm{ \sum_{i=1}^\npar A_i \otimes A_i } \leq \norm{ \sum_{i=1}^\npar A_i^2 }.
\end{equation}
\end{lemmaa}
\begin{proof}
We rewrite the initial matrix as follows
\begin{equation}
\sum_{i=1}^\npar A_i \otimes A_i = \sum_{i=1}^\npar (A_i \otimes \id ) (\id \otimes A_i ) = \tilde{A} \tilde{B},
\end{equation}
where $\tilde{A}=\begin{bmatrix} A_1 \otimes \id
\\
\vdots\\
A_\npar \otimes \id
\end{bmatrix}^\dag = \begin{bmatrix}
A_1 \dots A_\npar \end{bmatrix} \otimes \id $ and $\tilde{B}=\begin{bmatrix} \id \otimes  A_1
\\
\vdots\\
\id \otimes  A_\npar
\end{bmatrix} =
\id
\otimes \begin{bmatrix}   A_1
\\
\vdots\\
  A_p
\end{bmatrix} $.
Now we can use the submultiplicativity of the operator norm
\begin{equation}
\left \lVert \sum_{i=1}^\npar A_i \otimes A_i \right\rVert  = \left \lVert \tilde{A} \tilde{B} \right\rVert \leq \left \lVert \tilde{A}\right\rVert \left \lVert \tilde{B} \right\rVert
\end{equation}
and obtain the desired result by noticing that
$\left \lVert \tilde{A}\right\rVert = \sqrt{\left \lVert \tilde{A} \tilde{A}^\dag\right\rVert}=\sqrt{\left \lVert \sum_i^\npar A_i^2 \right\rVert}$
and
$\left \lVert \tilde{B}\right\rVert = \sqrt{\left \lVert \tilde{B}^\dag \tilde{B}\right\rVert}=\sqrt{\left \lVert \sum_i^\npar A_i^2 \right\rVert}$.
\end{proof}
We can now prove Theorem~\ref{theo:parallelCE}, essentially following~\cite{Fujiwara2008}.
\begin{proof}[Proof of Theorem~\ref{theo:parallelCE}]
The total entanglement-assisted channel QFI of the ensemble $ \{ \mathcal{E}_{\theta_x}^{\otimes N} \}$ is obtained using Theorem~\ref{theo:multiparFujiImaiGen}.
Being a minimization over Kraus operators if we restrict the minimization over operators with a particular form we obtain an upper bound.
The most natural choice for the channels $\mathcal{E}_{\theta_x}^{\otimes N}$ is the tensor product of Kraus operators, defined recursively
\begin{equation}
\tilde{K}^{(N+1)}_{\vec{k}_x} = \tilde{K}^{(N)}_{k_1} \otimes \tilde{K}^{(1)}_{k_2},
\end{equation}
where $\vec{k}_x \in \{1,\dots,r_x\}^N \times \{1,\dots,r_x\}$ is a multi-index and $\tilde{\vec{K}}^{(1)}_x= \tilde{\vec{K}}_x$ are the Kraus operator of the original channel.
Each multi-index $\vec{k}_x$ depends on $x$, yet for brevity we suppress the dependence on $x$ when expressing its components $k_i \coloneqq k_{x,i}$.
We introduce the quantities
\begin{gather}
\alpha_x^{(N)} = \sum_{\vec{k}_x} \partial_{\theta_x} \tilde{K}^{(N)\dag}_{\vec{k}_x} \partial_{\theta_x} \tilde{K}^{(N)}_{\vec{k}_x} \\ \beta_x^{(N)} = \sum_{\vec{k}_x} \partial_{\theta_x} \tilde{K}^{(N)\dag}_{\vec{k}_x} \tilde{K}^{(N)}_{\vec{k}_x},
\end{gather}
and $\alpha_x \coloneqq \alpha_x^{(1)}$ and $\beta_x \coloneqq \beta_x^{(1)}$ are the quantities that appear in the statement of the theorem.
Following the derivation of~\cite{Fujiwara2008} we obtain
\begin{widetext}
\begin{equation}
\sum_{x=1}^{\npar} \alpha_x = \sum_{x=1}^\npar \left( \sum_{\substack{i,j \\ i+j=n-1}} \id^{\otimes i} \otimes \alpha_x \otimes \id^{\otimes j} - 2 \sum_{\substack{i,j,k\\ i+j+k=n-2 }} \id^{\otimes i} \otimes \beta_x \otimes \id^{\otimes j} \otimes \beta_x \otimes \id^{\otimes k} \right),
\end{equation}
\end{widetext}
the only difference being the additional summation over $x$.
Thanks to the triangle inequality
\begin{equation}
\norm{ \sum_{x=1}^\npar \alpha_x^{(N)}} \leq N \norm{ \sum_{x=1}^\npar \alpha_x} + N(N-1) \norm{ \sum_{x=1}^\npar \beta_x \otimes \beta_x };
\end{equation}
and the second term is upper bounded using~\eqref{eq:normineqTensor}  to obtain~\eqref{eq:parallelCEbound}.
\end{proof}
This bound is asymptotically equivalent to the adaptive bound when Heisenberg scaling is not allowed.
From the triangle inequality for the operator norm we also see that
\begin{equation}
\label{eq:tria_ineq_parallel}
\begin{split}
4& \min_{\mathfrak{h}} \left\{ N \norm{ \sum_{x=1}^{\npar} \w_x \alpha_x } +  N ( N - 1) \norm{ \sum_{i=1}^\npar \w_x \beta_x^2 }\right\} \\
&\leq \sum_{x=1}^\npar \w_x 4 \min_{h_x} \left\{ N \norm{ \alpha_x } +  N ( N - 1) \norm{ \beta_x }^2 \right\},
\end{split}
\end{equation}
where the quantity on the right hand side of~\eqref{eq:tria_ineq_parallel} is the sum of the independent single-parameter bounds; this is again a trivial bound that does not take into account inherent incompatibility.

\section{Bound for a general Markovian noise model}
\label{app:MarkovianBound}
We consider a probe system evolving in time according to a Gorini-Kossakowski-Lindblad-Sudarshan (GKLS) master equation:
\begin{multline}
\label{eq:LindbladMultiNoise}
\frac{d \rho}{d t} = -\I \theta_x [ H_x ,\rho ] + \sum_{j=1}^{J_x}  L_{x,j} \rho L_{x,j}^{\dag} \\ - \frac{1}{2} \left( L_{x,j}^{\dag} L_{x,j} \rho + \rho L_{x,j}^{\dag} L_{x,j} \right),
\end{multline}
where the parameter dependence enters linearly in the Hamiltonian part.
We can derive a bound for the total QFI of the most general strategy, which includes the application of arbitrary fast and frequent control operations, by considering the channel $\mathcal{E}_{\theta_x,dt}$, obtained by integrating the master equation for a time $dt$, and taking the limit $dt \to 0$.
For each channel we choose the following $J_x + 1 $ Kraus operators that reproduce the dynamics up to first order in $dt$:
\begin{align}
\label{eq:KrausLind}
K_{x,0} &= \id - \left( \frac{1}{2} \vec{L}_x^\dag \vec{L}_x + \I \theta_x H_x \right) dt + O(dt^2)\\
K_{x,j} &= L_{x,j} \sqrt{dt} + O(dt^{\frac{3}{2}}) \quad j=1,\dots,J_x.
\end{align}
Given the structure of the Kraus operator the Hermitian matrices in the minimization are written in the following block form
\begin{equation}
h_x = \left[
\begin{array}{c|c}
h_x^{0} & \vec{h}_x^{\dag} \\
\hline
\vec{h}_x & \mathbbm{h}_x
\end{array}
 \right].
\end{equation}
With this choice we can follow exactly the same approach used in the single-parameter case~\cite{Demkowicz-Dobrzanski2017, Zhou2017} and fix the total probing time $T$ such that the scheme is equivalent to a sequential one with a discrete number of channel uses $N=T/dt$ and eventually take the limit $dt \to 0$.
The final result is obtained by using the infinitesimal Kraus operators~\eqref{eq:KrausLind} in the sequential bound~\eqref{eq:optCQFIz}, where the free parameter $z$ allows to get a meaningful result for $dt \to 0$.
In particular we focus only on the case in which Heisenberg scaling $T^2$ is not possible and obtain the following SQL bound
\begin{gather}
\label{eq:SQLboundLind}
\CQFI^T_{\vec{\w}} \leq T \CQFIbnd_\vec{\w}, \quad  \CQFIbnd_{\vec{\w}} \coloneqq 4 \min_{\substack{\mathfrak{h} \\ \{\beta_x^{(1)} =0\}}} \norm{ \sum_{x=1}^{\npar} \w_x \alpha_x^{(1)} }
\end{gather}
where $\alpha_x^{(1)}= \left( \vec{h}_x^{(\frac{1}{2})} \id + \mathbbm{h}_x^{(0)} \vec{L}_x \right)^\dag \left( \vec{h}_x^{(\frac{1}{2})} \id + \mathbbm{h}_x^{(0)} \vec{L}_x \right)$,
$\beta_x^{(1)}= H_{x} + h_x^{0(1)} \id + \vec{h}_x^{\dag(\frac{1}{2})} \vec{L}_x + \vec{L}_x^\dag \vec{h}_x^{(\frac{1}{2})} + \vec{L}_x^\dag \mathbbm{h}^{(0)}_x \vec{L}_x $ and the optimization variables are the set $\mathfrak{h} = \left\{ h^{0(1)}_x, \vec{h}^{(\frac{1}{2})}_x, \mathbbm{h}^{(0)}_x  \right\}_{x=1}^\npar$.
The subscript in brackets indicates the corresponding order in $dt$.

The conditions $\beta_x^{(1)}=0$ are known as ``Hamiltonian in the Lindblad span'' (HLS) conditions, since they are equivalent to
\begin{multline}
H_x \in \spn_{\mathbbm{R}} \bigl\{  \id, (L_{x,j})^{\mathrm{H}}, \I ( L_{x,j})^{\mathrm{AH}} , (L_{x,j}^\dag L_{x,j'} )^{\mathrm{H}}, \\ \I ( L_{x,j}^\dag L_{x,j'} )^{\mathrm{AH}} \bigr\}, \nonumber
\end{multline}
where $\mathrm{H}$ and $\mathrm{AH}$ denote the Hermitian and anti-Hermitian parts and the sets are known as the Lindblad spans~\cite{Sekatski2016,Demkowicz-Dobrzanski2017,Zhou2017,Zhou2020} of to the different GKLS master equations.
These conditions have to be violated for all parameters in order to preserve Heisenberg scaling of the total variance with error-correction~\cite{Gorecki2020}.

\section{Recovering previous results}
\label{app:previous}
Here we show that a couple of existing results that have been derived in a different framework can be recovered as purification-based bounds.

\subsection{Unitary parameters}
The situation for unitary parameters is particularly simple, since there is only one Kraus operator.
Kura and Ueda~\cite[Theorem 1]{Kura2017} have derived a general bound and here we show that the same result can be obtained from the purification-based definition of the QFI matrix.
\begin{corollary}
\label{cor:KuraUeda}
For noiseless multiparameter estimation, with a linear parameter encoding $U_{\parvec} =  e^{- \I H_{\vec{\theta}}} = e^{- \I \sum_{x=1}^\npar \theta_x G_x}$, with $\npar$ Hermitian generators $\{ G_x \}_{x=1}^\npar$, we have
\begin{equation}
\label{eq:noiselessKuraUeda}
\tr \QFI( U_{\parvec} \ket{\psi_0} ) \leq \CQFI  \leq 4 \left\lVert \sum_{x=1}^\npar G_{x}^2  \right\rVert.
\end{equation}
\end{corollary}
\begin{proof}
From Eq.~\eqref{eq:singleuseCEbound} for a unitary evolution we have
\begin{equation}
\tr \QFI ( U_{\parvec}  \ket{\psi_0} ) \leq 4 \min_{ \mathfrak{h} }  \norm{ \sum_{x=1}^\npar  \partial_x \tilde{U}_{\parvec}^\dag  \partial_x \tilde{U}_{\parvec} } \leq  4 \norm{ \sum_{x=1}^\npar  \partial_x U_{\vec{\theta}}^\dag  \partial_x U_{\parvec} },
\end{equation}
since the choice $h_x=0$ needs not be optimal.
We recall the formula for the derivative of a unitary~\cite{Baumgratz2015}
\begin{equation}
\partial_x U_{\parvec} = -\I U_{\parvec} \left( \int_0^1 \! d \alpha \, e^{\I \alpha H_{\parvec}}  G_x e^{-\I \alpha H_{\parvec}} \right).
\end{equation}
We can use the triangle inequality for the integral to get the upper bound
\begin{align}
& \norm{ \sum_{x=1}^\npar  \partial_x U_{\parvec}^\dag  \partial_x U_{\parvec}} \\
= &\norm{ \sum_{x=1}^\npar \left( \int_0^1 d \alpha e^{\I \alpha H_{\parvec} } G_x e^{-\I \alpha H_{\parvec} } \right)\left( \int_0^1 d \alpha' e^{\I \alpha' H_{\parvec} } G_x e^{-\I \alpha' H_{\parvec} } \right)   } \nonumber  \\
\leq  & \int_0^1 d \alpha \int_0^1 d \alpha' \norm{ \sum_{x=1}^\npar \left(  e^{\I \alpha H_{\parvec }} G_x e^{-\I \alpha H_{\parvec} } \right)\left(  e^{\I \alpha' H_{\parvec} } G_x e^{-\I \alpha' H_{\parvec} } \right) } \nonumber  \\
= & \int_0^1 d \alpha \int_0^1 d \alpha' \norm{  \sum_{x=1}^\npar G_x e^{-\I (\alpha-\alpha') H_{\parvec} } G_x  } \nonumber  \\
\leq & \int_0^1 d \alpha \int_0^1 d \alpha' \norm{  e^{-\I (\alpha-\alpha') H_{\parvec} } } \norm{ \sum_{x=1}^\npar G_x^2 \right\rVert  = \left\lVert \sum_{x=1}^\npar G_x^2  } , \nonumber
\end{align}
where we have used the inequality \eqref{eq:normineqDDM}.
\end{proof}
In particular, when the parameters are frequencies $\omega_i = \theta_i / t$ we get $\CQFI  \leq 4 t^2 \norm{ \sum_{x=1}^{\npar} G_{x}^2 }$ as in~\cite{Kura2017}.

\subsection{RLD channel bound}
\label{app:RLD}
Hayashi introduced a channel bound based on the right logarithmic derivative (RLD) QFI~\cite{Hayashi2011} for a single-parameter family of channels.
Further properties and connections with hypothesis testing have been explored in~\cite{Katariya2020a}.
While the applicability of this construction is more limited than general purification-based bounds~\cite{Demkowicz-Dobrzanski2012}, it has the advantage of being expressed only in terms of the Choi–Jamiołkowski (CJ) matrix of the channel.
Very recently, the same approach has been extended to multiple parameters in~\cite{Katariya2020b}.
Following the reasoning of~\cite{Kolodynski2014} we show that, when it is defined, the multiparameter RLD channel bound corresponds to a particular purification, thus being less tight than the optimal bound.

The CJ matrix of the channel $\mathcal{E}_{\parvec}: \mathcal{T}(\hilb_\mathrm{in} ) \to  \mathcal{T}(\hilb_\mathrm{out}) $ is an unnormalized state on the space $\mathcal{T}(\hilb_\mathrm{out} \otimes \hilb_A)$, where $\dim \hilb_A = d_{\mathrm{out}}$ defined as
\begin{equation}
\Omega_{\parvec} = \mathcal{E}_{\parvec} \otimes \idch (  | \id \rangle \langle \id | ) = \sum_{ij} \mathcal{E}_{\parvec} ( | {i} \rangle \langle  j |_S) \otimes | i \rangle \langle j|_A  ,
\end{equation}
where $\{ \ket{j}_S \}_{j=1}^{d_{\mathrm{in}}}$ and $\{ \ket{j}_A \}_{j=1}^{d_{\mathrm{out}}}$ are orthonormal bases of $\hilb_\mathrm{in}$ and $\hilb_A$, while $\ket{\id}=\sum_{i=1}^{d_\mathrm{in}} \ket{i}_S \ket{i}_A$ is an unnormalized maximally entangled state.
We use the compact notation $\ket{M}=\sum_{i,j=1}^{d_\mathrm{in}} \braket{i|M|j} \ket{i} \ket{j} = M \otimes \id \ket{\id} = \id \otimes M^T \ket{\id}$ that can be used to write the CJ matrix as $\Omega_{\parvec}=\sum_{i=1}^{l} | K_i \rangle \langle K_i | $ where the operators $K_i$ are an arbitrary Kraus decomposition of the channel.
The CJ matrix can be diagonalized as $\Omega_{\parvec}=\sum_{i=1}^{r} \lambda_i | \Psi_i \rangle \langle \Psi_i |$, where $\lambda_i>0$ and $\braket{\Psi_i|\Psi_j}=\delta_{ij}$ and this defines the canonical Kraus decomposition $\ket{K_i}=\sqrt{\lambda_i}\ket{\Psi_i}$.

The following bound for a general adaptive strategy with $N$ uses of the channel was obtained in~\cite{Katariya2020b} (we consider $W=\id$ without loss of generality, since it can be understood as a reparametrization)
\begin{gather}
\label{eq:RLDseqbound}
\Delta^2 \tilde{\parvec} \geq \frac{\npar^2}{ N \CQFIbnd^{\mathrm{R}} } \\
\CQFIbnd^{\mathrm{R}}=\norm{ \sum_{x=1}^\npar \Tr_A \left[ (\partial_x \Omega_{\parvec}) \Omega_{\parvec}^{-1} (\partial_x \Omega_{\parvec}) \right] },
\end{gather}
valid when the following finiteness condition holds
\begin{equation}
\label{eq:finiteRLD}
\sum_{x=1}^\npar \left( \partial_x \Omega_{\parvec} \right)^2 \Pi^{\perp}_{\Omega} = 0,
\end{equation}
where $\Pi^{\perp}_{\Omega}$ is the projector on the kernel of the CJ matrix $\Omega_{\parvec}$ and the inverse is taken on the support.
Otherwise the bound is trivial (diverging) when \eqref{eq:finiteRLD} is not satisfied.
In a moment, we will show that this condition is equivalent to the following conditions
\begin{multline}
\label{eq:locnonextr}
\partial_x K_i = \sum_{j} \nu_{x,ij} K_i \\
\implies \partial_x \Omega_{\parvec} = \sum_{ij} \mu_{x,ij} | K_i \rangle \langle K_j| \; \forall \, x = 1, \dots, \npar,
\end{multline}
where $\{ \nu_{x} \}$ are $\npar$ complex matrices, $\mu_{x} = \nu_{x} + \nu_x^{\dag}$ are twice their Hermitian part and $\{ K_i \}$ is the canonical Kraus representation.
This also means that the partial derivatives of the CJ matrix vanish outside the support of the CJ matrix.
For a single-parameter channel this condition is known as $\varphi$-non-extremality~\cite{Demkowicz-Dobrzanski2012,Kolodynski2014}.
When it is satisfied for all the parameters of a quantum channel we dub it \emph{local non-extremality}, or equivalently we say the channel to be locally non-extremal.
Now we proceed to show the equivalence of local non-extremality~\eqref{eq:locnonextr} and the finiteness condition~\eqref{eq:finiteRLD}.
The situation is essentially equivalent to the single-parameter case described in~\cite{Kolodynski2014}, modulo minor observations.

When the channel is locally non-extremal~\eqref{eq:locnonextr} the finiteness condition~\eqref{eq:finiteRLD} holds trivially, since $\bra{K_i} P_{\Omega}^{\perp} = 0 \; \forall \, i$.
On the other hand, we notice that the condition~\eqref{eq:finiteRLD} is equivalent to
\begin{equation}
\label{eq:PerpProjCond}
\Pi^{\perp}_{\Omega} \left[ \sum_{x=1}^\npar (\partial_x \Omega_{\parvec})^2 \right] \Pi^{\perp}_{\Omega}=0,
\end{equation}
since the matrix $\sum_{x=1}^\npar (\partial_x \Omega_{\parvec})^2$ is Hermitian.
We write the derivatives of the Kraus operators separating the components in the support of $\Omega_{\parvec}$ and those in the kernel: $\ket{\partial_x K_i} = \sum_{j=1}^r \nu_{x,ij} \ket{K_j} + \ket{L_{x,i}}$, where $\braket{K_i| L_{x,j}}=0 \; \forall \, i,j=1,\dots,r \, \forall\, x=1,\dots,\npar$ and $\{ \nu_{x} \}$ are $\npar$ complex matrices of dimension $r{\times}r$.
The condition~\eqref{eq:PerpProjCond} becomes
\begin{equation}
\sum_{x=1}^\npar  \left(  \sum_{i=1}^r |L_{x,i}\rangle \langle K_i  | \right) \left( \sum_{j=1}^r |K_j\rangle \langle L_{x,i}  | \right) = 0.
\end{equation}
This equality has the form $\sum_i A_i^\dag A_i =0$ and since $A_i^\dag A_i \geq 0$ we must have that $A_i = 0 \; \forall i$, furthermore since the vectors $\ket{K_i}$ are orthogonal we obtain $\ket{ L_{x,i}} =0 \; \forall i,x$, which means
\begin{equation}
\label{eq:locnonextr2}
\partial_x K_i = \sum_{j} \nu_{x,ij} K_i \qquad \forall \, x = 1, \dots, \npar.
\end{equation}
From this we have that
\begin{align}
& \partial_x \Omega_{\parvec} = \sum_{i=1}^r | \partial_x K_i \rangle \langle  K_i | + | K_i \rangle \langle  \partial_x  K_i |  \\
& = \sum_{i,j=1}^r \nu_{x,ij} | K_j \rangle \langle  K_i | + \nu_{x,ij}^{*} | K_i \rangle \langle  \partial_x  K_j | = \sum_{i,j=1}^r \mu_{x,ij} | K_j \rangle \langle  K_i |, \nonumber
\end{align}
where $\mu_{x} = \nu_x + \nu_{x}^{\dag}$ is twice the Hermitian part.

Now we can show that~\eqref{eq:locnonextr} implies the satisfaction of the HKS condition for all parameters and that the bound~\eqref{eq:RLDseqbound} has the same form $4\norm{ \sum_x \alpha_x }$ of purification-based bounds.
From the derivatives of the Kraus operators~\eqref{eq:locnonextr} we obtain the equality
\begin{align}
&\frac{\I}{2} \sum_i | K_i \rangle \langle \partial_x K_i | -  | \partial_x K_i \rangle \langle K_i | \\
 &= -\frac{\I}{2} \sum_{i,j} (\nu_x - \nu_{x\dag})_{ij} | K_i \rangle \langle K_j | = \sum_{i,j} (h^\mathrm{R}_{x})_{ij} | K_i \rangle \langle K_j |, \nonumber
\end{align}
where $h^\mathrm{R}_{x}=-\I \nu_{x}^{\mathrm{A}} =  -\frac{\I}{2}(\nu^x - \nu^{x\dag}) $ and by partial tracing over $\hilb_A$ we obtain the HKS conditions $\partial_x \vec{K}^\dag \vec{K}=-\I \vec{K}^\dag h^\mathrm{R}_x \vec{K} $, since $\partial_x\vec{K}^\dag \vec{K} = - \vec{K}^\dag \partial_x \vec{K} $.
Since $\braket{K_i|\Omega_{\parvec}^{-1}|K_j} = \delta_{ij} $ the operators appearing inside the summation in~\eqref{eq:RLDseqbound} can be written as follows
\begin{align}
\Tr_{A}\left[ (\partial_x \Omega_{\parvec}) \Omega_{\parvec}^{-1} (\partial_x \Omega_{\parvec})  \right] & =
\Tr_{A}\left[ \sum_{i,j} (\mu_{x}^2)_{ij} |K_j \rangle \langle K_i |\right] \nonumber \\
& = 4 \partial_x \tilde{\vec{K}}^\dag \partial_x \tilde{\vec{K}},
\end{align}
where $\partial\tilde{\vec{K}}=\partial_x \vec{K} - \I h^\mathrm{R}_x \vec{K}$, since $\partial_x \vec{K} = \nu_x \vec{K}$ from~\eqref{eq:locnonextr} and $\nu_x - \I h^{\mathrm{R}}_x = \frac{1}{2} \mu^x $ by definition.
Therefore we have shown that the bound~\eqref{eq:RLDseqbound} has the same form of the bound~\eqref{eq:SQLbound}, but is evaluated with a generally suboptimal choice $\{ h_x^{\mathrm{R}} \}_{x=1}^\npar$, implying that
\begin{equation}
\CQFIbnd \leq \CQFIbnd^{\mathrm{R}}.
\end{equation}

\section{Semidefinite programs for the bounds}
\label{app:SDP}
We introduce the following matrix
\begin{equation}
\vec{D} = \begin{bmatrix}
\sqrt{\w}_1 \left( \partial_{\theta_1} \vec{K}_1 -\I h_1 \vec{K}_1 \right) \\
\vdots \\
\sqrt{\w}_{\npar} \left( \partial_{\theta_\npar} \vec{K}_\npar -\I h_{\npar} \vec{K}_\npar \right)
\end{bmatrix},
\end{equation}
where the derivatives of the Kraus operator of each channel $\mathcal{E}_{\theta_x}$ are put in column.
This matrix has dimension $ \bar{d} {\times} d_{\mathrm{in}}$, where $\bar{d}=d_{\mathrm{out}} \sum_{x=1}^{\npar} r_x$, where $r_x$ is the number of Kraus operators of each channel.
When considering the multiparameter estimation scenario with a single channel with $r$ Kraus operators, we have $\bar{d}= \npar d_{\mathrm{out}} r$.

The bound for a single use of the channel~\eqref{eq:singleuseCEbound} can be rewritten as the following SDP
\begin{equation}
\label{eq:SDPtotalQFI}
\CQFI = 4 \min_{t,\{ h_x \}} t \\
\quad \text{subject to }
\begin{bmatrix}
t \id_{d_\mathrm{in}} & \vec{D}^\dag \\
\vec{D} & \id_{\bar{d}}
\end{bmatrix} \geq 0.
\end{equation}

The finite-$N$ bound~\eqref{eq:parallelCEbound} for the parallel strategy, i.e. $\CQFIbnd_N = 4 \min_{\mathfrak{h}} \left\{ \norm{ \sum_{x=1}^{\npar} \w_x \alpha_x } +  ( N - 1) \norm{ \sum_{i=1}^\npar \w_x \beta_x^2 }\right\}$, can be obtained similarly to the single-parameter case~\cite{Koodynski2013} as follows
\begin{equation}
\begin{split}
\label{eq:SDPtotalQFI_finiteN}
  \CQFIbnd_N = 4 & \min_{t,v,\{ h_x \}} \{ t + (N-1) v \}
  \\ & \text{subject to } 
  \begin{bmatrix}
  t \id_{d_\mathrm{in}} & \vec{D}^\dag \\
  \vec{D} & \id_{\bar{d}}
  \end{bmatrix} \geq 0 \; \;
  \begin{bmatrix}
    v \id_{d_\mathrm{in}} & \vec{B}^\dag \\
    \vec{B} & \id_{\bar{d}}
  \end{bmatrix} \geq 0,
\end{split}
\end{equation}
where we have introduced
\begin{equation}
  \vec{B} = 
  \begin{bmatrix}
    \sqrt{\w}_1  \left( \partial_{\theta_1}\!\vec{K}_1 - \I h_1 \vec{K}_1 \right)^\dag \vec{K}_1 \\
    \vdots \\ 
    \sqrt{\w}_\npar  \left( \partial_{\theta_\npar}\!\vec{K}_\npar - \I h_\npar \vec{K}_\npar \right)^\dag \vec{K}_\npar
  \end{bmatrix}.
\end{equation}

The asymptotic SQL bound~\eqref{eq:SQLbound} is obtained simply by imposing the additional linear constraints $\beta_x=0 \; \forall \, x =1,\dots \npar$ to the SDP~\eqref{eq:SDPtotalQFI}.

Finally, the asymptotic SQL bound for Markovian noise~\eqref{eq:SQLboundLind} is obtained in the same way, but using
\begin{equation}
\vec{D}_\mathrm{Mark} =\begin{bmatrix}
\sqrt{\w_1} \left(\vec{h}_1^{(\frac{1}{2})} \id + \mathbbm{h}_1^{(0)} \vec{L}_1 \right) \\
\vdots \\
\sqrt{\w_2} \left(\vec{h}_{\npar}^{(\frac{1}{2})} \id + \mathbbm{h}_{\npar}^{(0)} \vec{L}_\npar \right),
\end{bmatrix}
\end{equation}
instead of $\vec{D}$; here $d_{\mathrm{in}}=d_{\mathrm{out}}=d$, $\bar{d}=d \sum_{x=1}^{\npar} J_x$ (where $J_x$ is the number of collapse operators of each master equation) and the optimization runs over the Hermitian matrices $\left\{ h^{0(1)}_x, \vec{h}^{(\frac{1}{2})}_x, \mathbbm{h}^{(0)}_x  \right\}_{x=1}^\npar$ with the linear constraints $\beta_x^{(1)}=0$.

\section{Algorithm to find an optimal state and evaluate its QFI matrix}
\label{app:optistate}
The derivation of the analogous single-parameter algorithm~\cite{Zhou2019e,Zhou2020} relies on Sion's minimax theorem and remains unchanged for our multiparameter figure of merit.
The only difference is that now we have a collection of $\npar$ matrices $\mathfrak{h}=\{ h_x \}_{x=1}^\npar$ instead of just one.
Thus we can adapt the two-step procedure of~\cite[Appendix F]{Zhou2020}; the algorithm to find an optimal state goes as follows (for simplicity we introduce the operator $\bar{\alpha}=\sum_{x=1}^\npar q_x \alpha_x$).
\begin{enumerate}
\item Find a set of optimal Hermitian matrices $\mathfrak{h}^\star$ by solving the SDP~\eqref{eq:SDPtotalQFI}, such that the operator $\bar{\alpha}^\star =  \bar{\alpha} |_{\mathfrak{h}=\mathfrak{h}^\star} $ satisfies $\min_{\mathfrak{h}} \norm{ \bar{\alpha} } = \norm{\bar{\alpha}^\star} $.
\item The support of the optimal state $\rho^\star$ is the eigenspace of the largest eigenvalue of the operator $\bar{\alpha}^\star$ and $\forall x$ it satisfies the constraints
\begin{equation}
\label{eq:optistate_const}
\begin{split}
\Re \left\{ \Tr \left[ \rho^\star ( \I \vec{K}_x^\dag \Delta h_x ) ( \partial_x \vec{K}_x - \I h_x \vec{K}) \right] \right\}=0 \\
\forall \, \Delta h_x \in \mathbbm{C}_{r_x{ \times }r_x}, \; \Delta h_x = (\Delta h_x)^\dag.
\end{split}
\end{equation}
\end{enumerate}
The constraints~\eqref{eq:optistate_const} are linear constraints on $\rho^\star$ and in practice they are imposed by fixing a basis of $r_x{\times}r_x$ Hermitian matrices.
We remark that the optimal state is generally mixed when the largest eigenvalue of $\bar{\alpha}^\star$ has multiplicity greater than one.
When a mixed state is optimal it means that an optimal strategy is to use an extended channel and take advantage of entanglement with the auxiliary system, as it is clear from the proof of Theorem~\ref{theo:multiparFujiImaiGen}.

This algorithm allows one to find an optimal state attaining the total QFI, even when it corresponds to a random sensing scenario and the Kraus operators $\vec{K}_x$ pertain to different quantum channels.
If we work in the multiparameter estimation scenario there is only one vector $\vec{K}$ of Kraus operators; once an optimal state is found, the QFI matrix elements are evaluated as
\begin{equation}
\QFI_{xy} = 4 \Re \left\{ \Tr \left[ \rho^\star  ( \partial_x \vec{K} - \I h^\star_x \vec{K})^\dag ( \partial_y \vec{K} - \I h^\star_y \vec{K})\right] \right\},
\end{equation}
since the matrices $\mathfrak{h}^\star$ correspond to the optimal purification we can apply the purification-based definition~\eqref{eq:purifdefQFIM}.

The same algorithm can be adapted to the asymptotic case by solving the SDP with the constraints $\beta_x = 0$ and additionally imposing $\vec{K}_x^\dag \Delta h_x \vec{K}_x =0 \; \forall x$.
In practice, these constraints are imposed by using a basis of the nullspace of the map $h \mapsto \sum_{j,k=1}^{r_x} h_{jk} K_{x,j}^\dag K_{x,k}$ from $r_x{\times}r_x$ to $d_\mathrm{in}{\times }d_{\mathrm{in}}$ Hermitian matrices.

\section{Details on the evaluation of the bounds for quantum metrology applications}

In the following calculations we make a series of ansätze on the form of the optimal matrices $\mathfrak{h}$.
These are mostly inspired by the numerical solution and justified by symmetry arguments.
While the optimality of the presented solutions has been tested against numerical results, we remark that any allowed choice of $\mathfrak{h}$ provides a valid bound on the total QFI.
To ease the notation, in this section we move the parameter label $x$ of the matrices $\mathfrak{h}$ to the superscript when necessary.

\subsection{Generalized amplitude damping channel}
\label{app:GADchannel}
The generalized amplitude damping channel is a qubit channel with the following Kraus operators~\cite{nielsen2010quantum}
\begin{equation}
\begin{split}
K_0 & = \sqrt{1-\nu} \begin{bmatrix} 1 & 0 \\ 0 & \sqrt{1-\gamma} \end{bmatrix}\quad K_1 = \sqrt{1-\nu} \begin{bmatrix} 0 & \sqrt{\gamma} \\ 0 & 0 \end{bmatrix}  \\
K_2 &=\sqrt{\nu} \begin{bmatrix} \sqrt{1-\gamma} & 0 \\ 0 & 1 \end{bmatrix}  \quad K_3= \sqrt{\nu} \begin{bmatrix} 0 & 0  \\ \sqrt{\gamma} & 0 \end{bmatrix}.
\end{split}
\end{equation}
and we are interested in the estimation of both the parameters $\nu$ and $\gamma$.
The estimation of $\gamma$ was studied in great detail in~\cite{Fujiwara2004}.
The multiparameter problem was studied in~\cite{Katariya2020b} as an application of the RLD channel bound introduced in Appendix~\ref{app:RLD}.

For this model sequential or parallel strategies do not give any advantage and we observe $\CQFI = \CQFIbnd$ and $\mathfrak{I} = \mathfrak{I}_{\infty}$.
The only nonzero elements of the optimal purification matrices are $h^\nu_{02}=h^{\nu*}_{02}= \I A $ and $h^\gamma_{02}=h^{\gamma*}_{02} = \I B$, where $A$ and $B$ are real numbers, but we do not report the full details to find $A$ and $B$ as functions of $\gamma$ and $\nu$.
However, it is simple to check this statements numerically by solving the SDP, the code for this example can be found in~\cite{githubrepo}.
We mention that the optimal scheme is to use a probe state $\sqrt{a} \ket{00} + \sqrt{1-a} \ket{11}$, making use of an auxiliary system, i.e. considering the extended channel $\mathcal{E}_{\nu,\gamma}\otimes \idch$.
The optimal degree of entanglement $a$ with the auxiliary system depends on the parameter values via a rather complicated function.

Qualitatively, the probe incompatibility cost is an decreasing function of $\gamma$ for a fixed $\nu \neq \frac{1}{2}$.
Moreover the problem is symmetrical around the value $\nu = \frac{1}{2}$ and incompatibility decreases symmetrically as $\nu$ goes from the extremes $0$ and $1$ to $\frac{1}{2}$ and only for $\eta=\frac{1}{2}$ there is no probe incompatibility.

For this problem it is interesting to compare our result to the RLD bound, already evaluated in~\cite[Appendix F]{Katariya2020b}.
One can see immediately that the multiparameter RLD  is not tight for this problem and it does not detect any probe incompatibility, since it can be easily checked that
\begin{equation}
\CQFIbnd^{\mathrm{R}} = \mathfrak{B}^{\mathrm{R}}_{\gamma} + \mathfrak{B}^{\mathrm{R}}_{\nu}.
\end{equation}
To give an idea we report numerical results for a particular choice of parameters, $(\nu,\gamma) = (\frac{1}{4},\frac{1}{2})$ for which we obtain $\CQFIbnd^{\mathrm{R}}  \approx 10.67 > \CQFI_{\nu} + \CQFI_{\gamma}  \approx 4.72 > \CQFI \approx 3.84$.

\subsection{Hamiltonian tomography with erasure noise}
\label{app:HamTomErause}

\subsubsection{Diagonal generators (lossy multi-phase estimation)}
\label{app:lossmultiphaseDeriv}
With the Kraus operators~\eqref{eq:ErasureQuditKraus} we obtain
\begin{equation}
K_i^\dag K_j = \begin{cases}
(1-\eta) |i \rangle \langle j| \quad & i,j>0 \\
\eta \id_d                     \quad & j=i=0\\
\end{cases}
\end{equation}
and the $\beta_x = 0$ HKS conditions become
\begin{equation}
- | x \rangle\langle x | = h^x_{00} \eta \id + (1-\eta)\sum_{ij} h^x_{ij} |i \rangle \langle j|
\end{equation}
which entail
\begin{equation}
\begin{cases}
h^x_{00} \eta + h^x_{ii} (1-\eta) = 0 \quad & i \neq x \\
h^x_{00} \eta + h^x_{xx} (1-\eta) = -1 \\
h^x_{ij} (1-\eta) = 0 \quad & i \neq j
\end{cases}
\end{equation}
therefore, taking advantage of the symmetry of the problem, we can parametrize the matrices $h^x$ satisfying the constraint as follows
\begin{gather}
\label{eq:hmatsDiagGenLoss}
h^x_{00} = A \quad h^x_{xx} = -\frac{1+A\eta}{1-\eta} \\
 h^x_{ii} = -\frac{A \eta}{1-\eta} \quad h^x_{ij} = 0 \quad h^x_{0i} = c_i,
\end{gather}
but we simplify the calculation with the ansatz $c_i = 0$.
We use the simplified form~\eqref{eq:SQLnoiseafter} to compute the bound.
We obtain
\begin{gather}
H=\sum_{x=1}^\npar (h^x)^2 = \mathrm{diag}( X, Y, Y, \dots, Y )\\
 X= \npar A^2  \quad Y =  \frac{ (\npar-1) A^2 \eta^2 + (1+A\eta)^2}{(1-\eta)^2}
\end{gather}
and since $\sum_{x=1}^\npar G_x^2=\id$ we have the function
\begin{equation}
\label{eq:eigen1}
\norm{ -\id + \sum_{i,j=0}^d H_{ij} K_i^\dag K_j } = -1 + X \eta + (1-\eta)Y,
\end{equation}
which is minimized for $A=-1/\npar$, giving the bound $\CQFIbnd_\mathrm{diag}=\frac{\eta}{1-\eta}\frac{4 \npar }{(\npar-1)}$.

The previous calculation was obtained with $\npar=d \geq 2$, but we can repeat the same calculation with $d > \npar \geq 2$ and keep the same diagonal matrices $h^x$ in~\eqref{eq:hmatsDiagGenLoss} as before, obtaining now
\begin{gather}
H'=\sum_{x=1}^\npar (h^x)^2 = \mathrm{diag}( X, Y, Y, \dots, Y , Z \dots, Z)\\
Z = \frac{ \npar A^2 \eta^2}{(1-\eta)^2},
\end{gather}
where $Y$ is repeated $\npar$ times and $Z$ is repeated $d-\npar$ times.
The operator inside the norm now has one block $\left(-1 + X \eta + (1-\eta)Y\right) \id_\npar$ analogous to the previous one~\eqref{eq:eigen1} and another block $\left( Z (1-\eta) + \eta X \right) \id_{d-\npar}$, where $\id_\npar$ is the projector on the span of the first $\npar$ canonical vectors and $\id_{d-\npar}=\id - \id_{\npar}$.
The minimization of the operator norm produces the same result as in the previous case, since the optimal maximal eigenvalue always pertains to the first block.

\subsubsection{Off-diagonal generators}

We focus on the $d(d-1)/2$ real off-diagonal generators $G_{\mu \nu} = \frac{1}{2} \left( | \mu \rangle \langle \nu | + | \mu \rangle \langle \nu | \right)$ and we use the convention $\mu > \nu$.
The HKS conditions entail
\begin{equation}
-\frac{1}{2}  \left( |\mu \rangle \langle \nu | + |\nu \rangle \langle \mu |  \right) = h_{00}^{\mu \nu} \eta \id + (1-\eta) \sum_{i,j > 0} h^{\mu \nu}_{i j } | i \rangle \langle j |,
\end{equation}
for $i,j>0$ we have
\begin{equation}
\begin{cases}
h^{\mu \nu}_{i j} = -\frac{1}{2(1-\eta)} \left( \delta_{i \mu} \delta_{j \nu} + \delta_{i \nu} \delta_{j \mu} \right) \quad  &i \neq j \\
h^{\mu \nu}_{ii} = - \frac{\eta}{1-\eta} h^{\mu \nu}_{00} \quad &i=j\,.
\end{cases}
\end{equation}
We assume the following form
\begin{gather}
h^{\mu \nu}_{00} = A \quad h^{\mu \nu}_{i0} = 0 \, \forall i >0 \\
 h^{\mu \nu}_{ij} = -\frac{\eta}{1-\eta} A \delta_{ij} -\frac{1}{2(1-\eta)} \left( \delta_{i \mu} \delta_{j \nu} + \delta_{i \nu} \delta_{j \mu} \right),\nonumber
\end{gather}
from which we obtain
\begin{align}
& \sum_{\mu > \nu} (h^{\mu \nu})^2 = \\
& \frac{d(d-1) A^2 }{2} |0\rangle \langle 0| + \frac{d-1}{4 (1-\eta)^2} \left(1+2d\eta^2 A^2 \right) \sum_{i=1}^d | i \rangle \langle i | \nonumber\\
& + \frac{\eta A }{(1-\eta)^2}\sum_{\mu > \nu } \left( | \mu \rangle \langle \nu | + | \nu \rangle \langle \mu | \right)  \nonumber
\end{align}
This is a spherical model (according to the terminology of~\cite{Kura2017}): $\sum_{ \mu > \nu} (G_{\mu \nu})^2 = \frac{d - 1}{4} \id$ and the function to minimize in the bound~\eqref{eq:SQLnoiseafter} becomes
\begin{align}
\norm{ \frac{d-1}{4}\frac{\eta}{1-\eta} \id + \frac{d (d-1) A^2}{2} \frac{\eta}{1-\eta} \id + \frac{\eta A}{1-\eta}\sum_{i\neq j} |i \rangle \langle j| }  \nonumber
\end{align}
and the optimal choice is simply $A=0$ so the bound is $\CQFIbnd_\mathrm{real} = \frac{\eta}{1-\eta}(d-1)$.
For the imaginary off-diagonal elements we can repeat the same reasoning and we arrive at the same bound.

We notice that the three operators $\sum_x \alpha_x$ for the three submodels are all proportional to the identity and so they saturate the triangle inequality with equality as per~\eqref{eq:HamEstBnd}.

\subsection{Phase and loss}
\label{app:phaseloss}

First, we consider the two separate single-parameter problems.
For phase estimation we have the single-use bound $\CQFI_\varphi= \frac{4 \eta }{\left(\sqrt{\eta }+1\right)^2}$ corresponding to the optimal matrix $h_\varphi=\mathrm{diag}(-1+\frac{1}{1+\sqrt{\eta}},-1)$ and the asymptotic bound $\CQFIbnd_\varphi= \frac{4 \eta}{1-\eta}$ corresponding to the matrix $h_\varphi=\mathrm{diag}(0,\frac{-1}{1-\eta})$.
For the estimation of $\eta$ both bounds coincide and we have $\CQFI_\eta = \CQFIbnd_\eta=\frac{1}{\eta(1-\eta)}$ and the optimal $h_\eta=0$ means that the original Kraus operators are already the optimal purification.

The single-use incompatibility cost~\eqref{eq:probe_inc_cost} requires solving the following minimization
\begin{equation}
\min_{h_\phi, h_\eta} 4 \norm{ \frac{1}{\CQFI_\phi} \alpha_\phi + \frac{1}{\CQFI_\eta} \alpha_\eta  },
\end{equation}
and the optimal matrices are $h_\phi=\mathrm{diag}(-\frac{\eta ^{3/2}-\sqrt{2} \sqrt{\eta ^{3/2}+\eta }+\eta }{\eta ^{3/2}+\eta -\sqrt{\eta }-1},-1)$ and $h_\eta = 0$, from which we obtain~\eqref{eq:incompPhLoss}.

For the asymptotic incompatibility bound, the condition $\beta_\phi = 0$ constraints the matrix $h_\phi$ to be the same as in the single parameter case.
We obtain $4\alpha_\phi=\mathrm{diag}(\CQFIbnd_\phi,0)$ and $4\alpha_\eta=\mathrm{diag}(\CQFIbnd_\eta,0)$, from which we can clearly see that there is no probe incompatibility: $\mathfrak{I}_{\infty} = 1$ and $\CQFIbnd = \mathfrak{B}_{\phi} + \mathfrak{B}_{\eta}$.

\subsection{Phase and dephasing}
\label{app:phasedephasing}

This problem is particularly simple, since the optimal matrices are always identical to the single-parameter ones.
Regarding the parameter $\phi$ we have the single-use bound~\cite{Koodynski2013} $\CQFI_\varphi= \eta^2$ corresponding to: $h_\varphi= \frac{1}{2} \id + \frac{\sqrt{1-\eta^2}}{2}\sigma_x$ and the  asymptotic bound~\cite{Demkowicz-Dobrzanski2012} $\CQFIbnd_\varphi=\frac{\eta^2}{1-\eta^2}$ corresponding to $h_\varphi= \frac{1}{2} \id + \frac{1}{2 \sqrt{1-\eta^2}}\sigma_x$.
Analogously to loss estimation in the previous section, the Kraus operators are already optimal for estimating $\eta$, i.e. $h_\eta =0$, and we have $\CQFI_\eta = \CQFIbnd_\eta=\frac{1}{1-\eta^2}$.

For these optimal matrices we obtain $4 \alpha_\varphi= \CQFI_\phi \id$ for a single use (and analogously $4 \alpha_\varphi= \CQFIbnd_\varphi \id$ for the asymptotic case), and $4 \alpha_\eta = \CQFI_\eta \id $.
Therefore we see no probe incompatibility $\mathfrak{I} = \mathfrak{I}_{\infty} = 1$, $\CQFI = \CQFI_{\varphi} + \CQFI_{\eta}$ and $\CQFIbnd = \CQFIbnd_{\varphi} + \CQFIbnd_{\eta}$.

\subsection{Diagonal generators with qudit dephasing}
\label{app:quditdephasing}

We can use the following Kraus representation for the qudit dephasing channel
\begin{equation}
\label{eq:quditDephKrausAlt}
K_0 = \sqrt{ \eta  } \id, \quad K_j = \sqrt{1-\eta} |k \rangle \langle k | \;\; i=1,\dots,d,
\end{equation}
which is not minimal, since it has $d+1$ operators and the rank of the channel is $d$, but it is more convenient for the calculation.

We start by recalling the HKS $\beta_x = 0$ conditions
\begin{align}
 \I \sum_{l=0}^d ( \partial_x K_l^\dag ) K_l = - | x \rangle \langle x | = \sum_{i,j=1}^d h^{x}_{ij} K_i^\dag K_j;
\end{align}
the rhs becomes
\begin{multline}
\label{eq:rhshx}
h^x_{00} \eta \id + \sum_{j=1} (h^x_{0j} + h^x_{j0}) \sqrt{\eta(1-\eta)} |j \rangle \langle j | \\ + \sum_{k,j=1} h^x_{jk} (1-\eta) \delta_{jk} |k \rangle \langle k |
\end{multline}
and we assume that the matrices $h^{x}$ have real elements, obtaining
\begin{align}
h^x_{00} \eta \id + \sqrt{\eta(1-\eta)} \sum_{j=1} 2 h^x_{0j}  |j \rangle \langle j | + (1-\eta) \sum_{k=1} h^x_{kk} |k \rangle \langle k |.
\end{align}

Now, we make the following ansatz on the form of the matrices $h^{x}$
\begin{equation}
\begin{cases}
h^x_{00} = A \\
h^x_{0x} = h^x_{x0} = C \\
h^x_{0j} = h^x_{0j} = B \qquad & j > 0 \land j \neq x \\
h^x_{xx} = F \\
h^x_{xj} = h^x_{jx} = G\qquad  & j > 0 \land j \neq x\\
h^x_{jj} = D \qquad &  i > 0 \land j \neq x  \\
h^x_{ij} = h^x_{ji} = E \qquad & i,j > 0 \land i,j \neq x \, ;
\end{cases}
\end{equation}
so that all the $d$ matrices $\left\{ h^x \right\}_{x=1}^d$ are parametrized by seven real parameters.
More explicitly they look like this
\begin{equation}
h^x = \left[ \begin{array}{c| c c c c c c}
A & B & B & \dots  & C & \dots & B \\
\hline
B & D & E &   & G &  & E \\
B & E & D &   & G & \dots & E \\
\vdots  &   &   &  \ddots &   &  &  \\
C & G & G &   & F &  & G \\
\vdots &   &   & \vdots  &   & \ddots  &   \\
B & E & E &  & G &  & D\\
\end{array}
\right],
\end{equation}
where the column and row that stand out are the $x$-th ones (starting to count from 0).
We have numerical evidence that matrices in this form attain the minimum.
With this simplification \eqref{eq:rhshx} becomes
\begin{multline}
A \eta \id + 2 B \sqrt{\eta(1-\eta)} \left( \id - |x \rangle \langle x | \right)  \\ +2 C \sqrt{\eta(1-\eta)} |x \rangle \langle x |
+ D (1-\eta) \left( \id - |x \rangle \langle x | \right) \\ + F (1-\eta) |x \rangle \langle x |,
\end{multline}
and the condition $\beta_x = 0 \,\, \forall x$ are satisfied iff
\begin{equation}
\begin{cases}
A \eta + 2 B \sqrt{\eta (1-\eta)} + D (1-\eta) = 0 \\
A \eta + 2 C \sqrt{\eta(1-\eta)} + F (1-\eta) = -1,
\end{cases}
\end{equation}
from which we can eliminate two variables, i.e.
\begin{gather}
\label{eq:FDsols}
F=-\frac{1 + A \eta  + 2 C  \sqrt{(1 - \eta) \eta}}{1-\eta}\\
D = -\frac{A \eta  + 2 B  \sqrt{(1 - \eta) \eta}}{1-\eta}.
\end{gather}
We obtain
\begin{equation}
H=\sum_{x=1}^d (h^x)^2 =
\left[ \begin{array}{c| c c c c}
X & Z & Z & \dots & Z \\
\hline
Z & Y & T &   & T \\
Z & T & Y &   & T  \\
\vdots & & & \ddots &  \\
Z & T & T &   & Y \\
\end{array}
\right],\\
\end{equation}
where
\begin{align}
    X =& d \left[ A^2 + (d - 1) B^2 + C^2\right] \\
    Y =&  (d - 1)\left[D^2 + B^2 + (d-2) E^2 + G^2 \right] \\
       &+ C^2 + F^2 + (d-1) G^2 \nonumber  \\
    Z =& (d - 1)\left[ B(A + D) + (d-2) B E + G C \right] \\
       & + C (A + F) + B G (d-1). \nonumber
\end{align}
Finally, the bound~\eqref{eq:SQLnoiseafter} amounts to the following minimization
\begin{equation}
\min_{A,B,C,E,G} \left[ X \eta + 2 Z  \sqrt{(1-\eta)\eta}+ (1-\eta)Y -1 \right],
\end{equation}
from which we obtain
\begin{equation}
\CQFIbnd = \frac{ 4 (d-1) \eta^2}{(2+ d \eta)(1-\eta)} = \frac{4 \eta}{1-\eta} \frac{d-1}{d+ \frac{2}{\eta}}.
\end{equation}

\bibliography{2021PurificationMultipar}
\end{document}